\documentclass[journal,12pt,onecolumn]{IEEEtran}
\usepackage{amsmath,amsfonts,amssymb,bm}
\usepackage{algorithm,algorithmic}
\usepackage{array}
\usepackage{textcomp}
\usepackage{stfloats}
\usepackage{verbatim}
\usepackage{graphicx}
\graphicspath{{figs/}}
\usepackage{subfigure}
\usepackage[T1]{fontenc}
\usepackage[utf8]{inputenc}
\usepackage{color}
\usepackage{colortbl}
\usepackage{makecell}
\usepackage[normalem]{ulem}
\usepackage{nomencl}
\usepackage{multirow}
\usepackage{booktabs,tabularx}
\usepackage{chngcntr}
\usepackage[hyphens]{url}  
\usepackage[hidelinks]{hyperref}
\usepackage{balance}

\counterwithin{equation}{section}

\newcommand{\myproj}{{SpecHDC}}
\newtheorem{definition}{Definition}

\newtheorem{theorem}{Theorem}[section]
\newtheorem{lemma}{Lemma}[section]

\newtheorem{rem}{Remark}[section]
\newtheorem{prop}{Proposition}[section]

\newcommand{\qed}{$\blacksquare$}
\def\hindu{\arabic}

\renewcommand{\theequation}{\hindu{section}.\hindu{equation}}

\newcommand\gattha[2]{\genfrac{}{}{0pt}{}{#1}{#2}}

\def\RR{{\mathbb R}}

\def\ZZ{{\mathbb Z}}

\def\bs#1{{\boldsymbol{#1}}}

\def\esssup{\mathop{\hbox{\textrm{ess sup}}}}

\def\be{\begin{equation}}
\def\ee{\end{equation}}
\def\bea{\begin{eqnarray}}
\def\eea{\end{eqnarray}}

\def\disp{\displaystyle}

\def\binom#1#2{\small{\left(\!\!\begin{array}{c}{#1}\\{#2}
\end{array}\!\!\right)}}
\def\donchitre#1#2{\vskip 6.5cm\noindent
\parbox[t]{1in}{\special{eps:#1.eps x=6.5cm y=5.5cm}}
\hbox to 7cm{}\parbox[t]{0.0cm}{\special{eps:#2.eps x=6.5cm y=5.5cm}}}

\def\span{\mbox{\mathsf{span }}}
\def\bs#1{{\boldsymbol{#1}}}

\def\gs{\gtrsim}
\def\ls{\lesssim}

\def\BibTeX{{\rm B\kern-.05em{\sc i\kern-.025em b}\kern-.08em
    T\kern-.1667em\lower.7ex\hbox{E}\kern-.125emX}}

\begin{document}
\title{\LARGE Hierarchical Clustering and Signal Denoising on Digraphs}
\author{Yi Wang, Sippanon Kitimoon, Hrushikesh N. Mhaskar, Xiaosheng Zhuang
\thanks{The research of H N. Mhaskar was supported in part by ONR grants N00014-23-1-2394, N00014-23-1-2790. The work of X. Zhuang was supported in part by the Research Grants Council of Hong Kong (Project No. CityU 11301224 and CityU 11300825). 
}
\thanks{Y. Wang and X. Zhuang are with the Department of Mathematics, City University of Hong Kong, Hong
Kong, SAR China. (email:ywan72@cityu.edu.hk, xzhuang7@cityu.edu.hk); S. Kitimoon is with the Data Science Research Center, Faculty of Science, Chiang Mai University, Chiang Mai 50200, Thailand (email: sippanon.kitimoon@cmu.ac.th); H. N. Mhaskar is with the Institute of Mathematical Sciences, Claremont Graduate University, Claremont, CA 91711, USA. (email:hrushikesh.mhaskar@cgu.edu)}
}


\maketitle

\begin{abstract}
In this paper, we propose a representation of a digraph (directed graph) as a Hermitian matrix derived from its adjacency matrix. This representation characterizes both the connectivity and the edge orientation of the digraph. Based on the spectral decomposition of the Hermitian matrix, a digraph clustering algorithm with $k$-means is introduced to produce a partition on the graph. Applying this algorithm (bottom-up) recursively to a digraph with partially labeled vertices yields a spectral hierarchical digraph clustering (\myproj) algorithm that produces consistent nested partitions of the digraph, or equivalently, a tree structure. Furthermore, based on the in-degree and out-degree of each cluster in the digraph clustering, a pair of hierarchical interval partitions (filtrations) can be derived  in a top-down manner to produce a pair of nested knot sequences. These knot sequences facilitate the construction of multilevel spline quasi-interpolants, enabling a
noisy graph signal to be decomposed into a coarse approximation and inter-level details, followed by adaptive thresholding and reconstruction. Experiments on synthetic and real-world digraphs demonstrate the superiority of our {\myproj} algorithm for digraph clustering across diverse graph structural properties (homophily and heterophily) and supervision settings. Moreover, experiments on digraph signal processing using multilevel spline quasi-interpolants further demonstrate the effectiveness of signal recovery on digraphs in terms of RMSE and SNR.
\end{abstract}


\section{Introduction}\label{sec:introduction}


\IEEEPARstart{S}{ignals} defined on graphs arise in a broad range of applications \cite{gama2019convolutional,sandryhaila2013discrete,ortega2018graph,leus2023graph}, including social, citation, web, sensor, and biological networks, where observations are associated with vertices and their dependencies are encoded by edges. Graph signal processing (GSP) extends classical signal processing from regular Euclidean domains to graph domains. In GSP, graph operators—most notably the graph Laplacian, adjacency matrix, and graph-shift operator— are used to define graph-domain notions of frequency and signal variation and to support fundamental tasks such as filtering \cite{shuman2020localized,wei2025vertex}, sampling \cite{tanaka2020sampling,tanaka2020generalized}, and recovery \cite{jung2019localized,jung2019semi,rey2023robust}.

Clustering can be viewed as a special form of graph signal recovery in GSP, where cluster labels constitute a discrete graph signal whose recovery reveals latent structures in graph-structured data. Early graph clustering methods formulated this task through combinatorial objectives such as graph cuts and
normalized cuts \cite{shi2000normalized,flake2004graph}. Spectral clustering subsequently relaxed these discrete objectives and embedded vertices using selected eigenvectors of the graph Laplacian, followed by $k$-means or related partitioning procedures \cite{ng2001spectral,von2007tutorial}. From a GSP perspective, these eigenvectors represent slowly varying structural modes, while the resulting cluster assignments constitute a piecewise-coherent graph signal. These classical formulations are primarily developed for undirected graphs, whereas many real-world networks are inherently directed (i.e., digraphs).  In digraphs, edge orientations encode asymmetric relations such as influence, transition, migration, and information flow, making directionality an essential component of the clustering structure. Existing digraph clustering methods address such
asymmetry from several perspectives. Symmetrization-based methods construct undirected similarities from directed connectivity patterns \cite{satuluri2011symmetrizations,chui2018representation}, while source--target methods distinguish the sending and receiving roles of vertices. For example, DI-SIM uses the left and right singular vectors of a regularized directed Laplacian to identify
asymmetric block structures \cite{rohe2016co}. Random-walk and flow-based approaches instead characterize clusters through directed transition behavior, including stationary-distribution-based Laplacians \cite{chung2005laplacians,zhou2005learning} and description-length objectives \cite{rosvall2008maps}. More recently, Hermitian spectral methods have encoded edge orientation in complex-valued operators while retaining real eigenvalues and an orthonormal eigenbasis \cite{cucuringu2020hermitian}; related skew-symmetric formulations reduce complex-domain computation while preserving the relevant directed-cut structure \cite{hayashi2022skew}.

On the other hand, it is well known that spline approximation provides an efficient and effective tool in classical signal processing. Splines are piecewise polynomials over a knot sequence that can be evaluated fast \cite{deboor2001practical}. 
Moreover, in contrast to exact interpolation, spline quasi-interpolation determines approximation coefficients from local samples or bounded functionals, avoiding a global interpolation system while retaining polynomial reproduction and favorable approximation orders \cite{de1973spline,chuidiamond90}. Spline constructions have also been closely connected with multiresolution representations, for example, spline--wavelet filter banks provide efficient multirate decomposition and reconstruction
\cite{chui2015representation,chui2016multirate}, while hierarchical and truncated hierarchical B-splines enable adaptive local refinement through nested spline spaces \cite{speleers2017hierarchical,bracco2016bivariate,kunoth2018foundations}. 



These two research threads suggest a natural but unexplored connection. Multi-level digraph clustering reveals how fine-grained vertex groups are progressively merged into coarse structural units, while spline approximation requires an ordered knot configuration that determines local support and resolution. Establishing this connection poses two coupled challenges. First, recursive clustering must preserve asymmetric inter-cluster interactions, incorporate limited supervision, and maintain consistent relations across successive coarse digraphs. Second, the discrete partitions and their level-wise directional
connectivity must be transformed into ordered intervals suitable for spline construction. These considerations lead to two closely related questions:
\begin{itemize}
    \item \textbf{How can digraphs be clustered consistently across levels with limited supervision?}
    \item \textbf{How can multi-level clusters induce valid knot sequences for spline approximation?}
\end{itemize}
The first challenge requires a direction-aware spectral representation that incorporates available labels. The second requires an interval construction that combines the partition
and directed adjacency information at every level. Addressing them jointly enables the resulting spline approximation to adapt to the hierarchical and asymmetric structures revealed by digraph clustering.

To address these challenges, we propose {\myproj}, a semi-supervised spectral hierarchical digraph clustering algorithm, together with a hierarchy-induced spline construction. {\myproj} is based on a generalized Hermitian matrix that jointly encodes reciprocal connectivity and asymmetric edge orientation, yielding a direction-sensitive spectral representation with real eigenvalues and an orthonormal eigenbasis. Unlike single-level methods that independently compute partitions at prescribed granularities, {\myproj} constructs a consistent hierarchy through recursive digraph coarsening. At each level, the current clusters are contracted into coarse vertices, and the directed adjacency weight between two coarse vertices is obtained by aggregating the edge weights between their constituent clusters. A new generalized Hermitian matrix is then constructed from the resulting coarse adjacency matrix, and its spectral representation is used to determine the partition at the next level. Repeating this procedure yields a sequence of nested partitions and their corresponding coarse directed graphs, with the available labels anchoring the finest-level partition and propagating through the hierarchical construction. Based on the partition and adjacency matrix obtained at each level, we further derive in-degree and out-degree intervals that characterize the directional connectivity of the corresponding clusters. The intervals collected across the hierarchy are organized into multi-level knot sequences for spline construction. Consequently, the resulting splines incorporate both the hierarchical organization of the clusters and the asymmetric inflow and outflow patterns of the digraph, providing a direction-aware and structure-adaptive representation beyond the discrete clustering results. Empirically, controlled experiments on synthetic homophilic and heterophilic digraphs characterize the behavior of {\myproj} under different directed structural regimes. Comparisons with representative clustering baselines on five labeled and two unlabeled real-world networks further demonstrate its practical effectiveness, while additional denoising experiments validate the utility of the hierarchy-induced in-degree and out-degree knot sequences for structure-adaptive signal recovery.

\subsection{Related work}

\noindent\textbf{Digraph Clustering.}
From a graph signal processing perspective, digraph clustering seeks a direction-aware graph operator whose spectral representation supports the recovery of cluster-label signals over asymmetric topologies. Existing methods mainly model directionality through symmetric transformations, source--target representations \cite{satuluri2011symmetrizations,rohe2016co,zhang2022directed}, random-walk dynamics \cite{chung2005laplacians,rosvall2008maps}, and Hermitian operators \cite{cucuringu2020hermitian,hayashi2022skew}. 

Satuluri and Parthasarathy~\cite{satuluri2011symmetrizations} constructed symmetric similarity matrices from products involving the directed adjacency matrix and its transpose, enabling conventional spectral clustering through common-predecessor and common-successor relations. Rohe et al.~\cite{rohe2016co} proposed DI-SIM, which uses the left and right singular vectors of a regularized directed Laplacian to characterize the sending and receiving roles of vertices; related methods similarly derive source- and target-side representations from incoming and outgoing
connectivity patterns \cite{zhang2022directed}. Random-walk-based methods instead identify clusters through transition behavior. Chung's directed Laplacian constructs a spectral operator from a random-walk transition matrix and its stationary distribution \cite{chung2005laplacians}, and was subsequently applied to learning from labeled and unlabeled digraph data \cite{zhou2005learning}. InfoMap~\cite{rosvall2008maps} detects modules by minimizing the description length of a random walk. Hermitian formulations provide a direction-preserving spectral alternative. Cucuringu et al.~\cite{cucuringu2020hermitian} encoded edge orientation in a complex Hermitian matrix with real eigenvalues and an orthonormal eigenbasis, allowing conventional spectral embedding to retain asymmetric information. Hayashi et al.~\cite{hayashi2022skew} further related this representation to
a real skew-symmetric formulation, reducing complex-domain computation while preserving the relevant directed-cut structure. Other studies have explored flow imbalance and higher-order relations among directed clusters \cite{he2022digrac,laenen2020higher}. 

Despite these advances, most existing methods produce single-level partitions at prescribed granularities. Direction-aware spectral hierarchical clustering remains less explored, particularly when partial labels must be incorporated and asymmetric inter-cluster interactions must be preserved during recursive graph coarsening.

\noindent\textbf{Spline Approximation. }
Spline approximation \cite{dahmen1980multidimensional, deboor2001practical, cox2006practical, hasan2024b} is a classical framework for representing, interpolating, and smoothing functions through piecewise-polynomial basis functions defined over a knot sequence. The compact support and controllable smoothness of B-splines make them particularly suitable for local approximation, geometric modeling, numerical computation, and signal reconstruction \cite{deboor2001practical}. 
Unlike classical interpolation, which enforces exact sample reproduction and often requires solving a global system, spline quasi-interpolation constructs coefficients from local function values or bounded linear functionals, providing an efficient and stable approximation while retaining locality and polynomial reproduction. Fundamental approximation properties of spline quasi-interpolants were established in \cite{de1973spline}, with further developments for scattered-data approximation given in \cite{quasiint2000}. 
Spline methods are also closely related to multiresolution and adaptive representation. Spline--wavelet constructions support efficient multiscale decomposition \cite{chui2016multirate}, while hierarchical spline spaces enable local refinement through nested approximation spaces \cite{speleers2017hierarchical,bracco2016bivariate}. These methods rely strongly on the underlying knot sequence or refinement hierarchy, which is typically specified in advance. Although variational splines extend interpolation and recovery to graph-supported signals through graph Laplacians \cite{pesenson2009variational}, they do not explicitly derive spline knots from the directional and hierarchical structure of a digraph. In this work, the in-degree and out-degree intervals obtained at successive levels of digraph clustering are assembled into multi-level knot sequences, allowing the quasi-interpolatory spline to adapt its resolution to both the hierarchical organization and asymmetric connectivity of the underlying digraph.

\subsection{Contributions}

The main contributions of this work are summarized as follows:
\begin{itemize}
    \item We introduce a generalized Hermitian matrix that jointly models reciprocal connectivity and asymmetric edge orientation, providing a flexible direction-sensitive spectral representation with real eigenvalues and an orthonormal eigenbasis.
    \item We develop SpecHDC, which recursively constructs coarse directed graphs by aggregating inter-cluster edge weights, recomputes the generalized Hermitian representation at each level, and produces nested partitions while incorporating the available label information.
    \item We derive level-wise in-degree and out-degree intervals from the clustering results and their corresponding adjacency matrices and organize these intervals into multi-level knot sequences, enabling direction-aware and structure-adaptive spline approximation.
    \item Since the knots are not equidistant, the de Boor inequality yielding stability bounds for splines does not hold. We obtain the necessary stability bounds with respect to a doubling measure. (See Section~\ref{sec:splines} for details). 
    \item Controlled experiments on synthetic homophilic and heterophilic digraphs analyze the structural behavior of \myproj; comparisons on five labeled and two unlabeled real-world networks demonstrate its clustering effectiveness; and denoising experiments confirm the value of the proposed knot sequences for graph signal recovery.
\end{itemize}

\subsection{Structure of the paper} 
The remainder of this paper is organized as follows. Section~\ref{sec:digraph_clustering} presents the notation and preliminaries for digraph clustering, develops the proposed spectral hierarchical clustering algorithm, and provides illustrative examples. Section~\ref{sec:experiments} evaluates the clustering performance on synthetic directed stochastic block models and real-world digraphs across diverse graph structural properties (homophily and heterophily).
Section~\ref{sec:splines} introduces splines with general measures and establishes the stability of the proposed quasi-interpolation scheme.
Section~\ref{sec:digraph_signal_processing} connects the obtained digraph hierarchy with graph signal processing by constructing multi-level spline representations and examining their performance in signal denoising. Finally, Section~\ref{sec:conclusions} concludes the paper.


\section{Digraph Clustering}\label{sec:digraph_clustering}

In this section, we introduce our \myproj~algorithm based on a proposed generalized Hermitian matrix and present its induced hierarchical interval partitions. Some synthetic examples are used to demonstrate our approach. 

\subsection{Notation and Preliminary on Graph Clustering}
\noindent\textbf{Undirected and Directed Graphs.}
Let $G=(V,E,A)$ be a weighted graph with vertex set
$V=\{v_1,\ldots,v_N\}$, edge set $E\subset V\times V$, and weighted adjacency matrix
$A=[a_{ij}]_{1\le i,j\le N}\in\mathbb{R}_{+}^{N\times N}$, where $a_{ii}=0$ and $a_{i,j}\neq 0$ if $(v_i,v_j)\in E$. For an undirected graph, $A=A^\top$,  the degree of $v_i$ is $d_i=\sum_{j=1}^{N}a_{ij}$, and the degree matrix
$D=\operatorname{diag}(d_1,\ldots,d_N)$ is diagonal. For a digraph, $a_{ij}>0$ denotes an edge from $v_i$ to $v_j$, and generally $a_{ij}\neq a_{ji}$. The out-degree and in-degree of $v_i$ are given by
\begin{equation}
d_i^{\mathrm{out}}=\sum_{j=1}^{N}a_{ij},\qquad d_i^{\mathrm{in}}=\sum_{j=1}^{N}a_{ji},
\label{eq:directed_degrees}
\end{equation}
and with the degree matrices $D_{\mathrm{out}}=\operatorname{diag}(d_1^{\mathrm{out}},\ldots,d_N^{\mathrm{out}})$ and $D_{\mathrm{in}}=\operatorname{diag}(d_1^{\mathrm{in}},\ldots,d_N^{\mathrm{in}})$. The graph volume is
\begin{equation}
\operatorname{vol}(G)=\sum_{i=1}^{N}d_i^{\mathrm{out}}=\sum_{i=1}^{N}d_i^{\mathrm{in}}=\sum_{i,j=1}^{N}a_{ij}.
\label{eq:graph_volume}
\end{equation}
A $k$-way clustering is a partition
$\mathcal{C}=\{C_1,\ldots,C_k\}$ of $V$ so that $V=\cup_{j=1}^k C_j$, equivalently represented by the label signal $y:V\rightarrow\{1,\ldots,k\}$ with $C_j=y^{-1}(j)$. In the semi-supervised setting, labels are known only for vertices indexed by $\mathcal{I}_\mathrm{knw}\subseteq\{1,\ldots,N\}$ so that  $Y_{\mathrm{knw}}(i)\in\{1,\ldots,k\}$ is given in advance for $i\in \mathcal{I}_\mathrm{knw}$. If $\mathcal{I}_{\mathrm{knw}}=\emptyset$, then it corresponds to the unsupervised setting. The remaining assignments are inferred subject to these known labels.

\noindent\textbf{$k$-Means.} Given $N$ observations $z_1, z_2,\ldots, z_N\in\mathbb R^n$ of $n$-dimensional vectors associated with vertices $v_1,\ldots, v_N$ of $V$, respectively, $k$-means clustering \cite{jain2010data} aims to partition $V$ into a $k$-way clustering $\mathcal C = \{C_1, C_2, \ldots, C_k\}$ of $V$, equivalently, to obtain a label signal $y:V\rightarrow \{1,\ldots,k\}$, so as to minimize the within-cluster sum of squares. Formally, the objective function is given by
\begin{equation}\label{k-means}
\min_{\mathcal C=\{C_j=y^{-1}(j)\}_{j=1}^k}\sum_{j=1}^{k}\sum_{v_i\in C_j}
\left\|z_i-c_j\right\|_2^2,\mbox{ s.t. } y|_{\mathcal I_{\mathrm{knw}}}=Y_{\mathrm{knw}},
\end{equation}
where $c_j=\frac{1}{|C_j|}\sum_{v_i\in C_j}z_i$ is the centroid of the vertices in $C_j$. See {\it Supplementary Material Part A.6} for more details. For a graph, the set $\{z_1, z_2,\ldots, z_N\}$ is typically obtained via spectral embedding deduced from
a graph operator, e.g., the graph Laplacian \cite{von2007tutorial}. For digraphs, defining such a graph operator, reflecting both the connectivity and the edge orientation of the digraph, is key to accurate embedding for the subsequent $ k$-means algorithm.

\subsection{Digraph Clustering}
To preserve both reciprocal connectivity and asymmetric edge orientation, we propose a generalized Hermitian representation from the weighted adjacency matrix $A$. Specifically, the generalized Hermitian matrix is defined as
\begin{equation}
W=\alpha(A+A^\top)+\beta(A-A^\top)\mathrm{i},
\end{equation}
where $\mathrm{i}=\sqrt{-1}$, $\alpha\in[0,1]$, and
$\beta\in(0,1]$. The symmetric component $A+A^\top$ characterizes the overall connectivity between two vertices, whereas the skew-symmetric component $A-A^\top$ captures the directional imbalance between opposite edges. Thus, $\alpha$ and $\beta$ control the contributions of reciprocal connectivity and edge directionality, respectively. Since $A+A^\top$ is symmetric and $A-A^\top$ is skew-symmetric, we have $W^{*}=\alpha(A+A^\top)^\top-\beta(A-A^\top)^\top\mathrm{i}=W$,
where $(\cdot)^{*}$ denotes the complex conjugate transpose. Therefore, $W$ is Hermitian and admits an eigendecomposition 
\begin{equation}
W = U\Lambda U^{*},
\end{equation}
where
$U=[g_1,\ldots,g_N]\in \mathbb C^{N\times N}$
contains an orthonormal set of eigenvectors, and
$\Lambda=\operatorname{diag}(\lambda_1,\ldots,\lambda_N)$
contains the corresponding real eigenvalues.

\begin{rem}
Satuluri and Parthasarathy in \cite{satuluri2011symmetrizations} used the bibliometric symmetrization $W= AA^\top + A^\top A$ for digraph representation. Chui et al. in \cite{chui2018representation} jointly use the bibliographic coupling matrix $AA^\top$ and the co-citation strength matrix $A^\top A$ as a pair $(AA^\top,A^\top A)$ to explore digraph signal representation as a 2-dimensional embedding. Cucuringu et al. in \cite{cucuringu2020hermitian} proposed the Hermitian adjacency matrix $W= (A-A^\top)\mathrm{i}$.  Hayashi et al.~\cite{hayashi2022skew} focus on the real skew-symmetric part $A-A^\top$. Despite the advantages of symmetrization, most of the above methods lack a complete characterization of a digraph in terms of both connectivity and directional orientation.  Our formulation of the generalized Hermitian matrix $W = \alpha (A+A^\top) + \beta (A-A^\top) \mathrm{i}$ not only characterizes both connectivity $A+A^\top$ and directional orientation $(A-A^\top)\mathrm{i}$, but also provides weights between the two characteristics, thereby enhancing the representation power of the a digraph for subsequent tasks. 
\end{rem}

Given a threshold $\epsilon>0$, we retain the eigenvectors whose
eigenvalues satisfy $|\lambda_j|>\epsilon$.
The corresponding spectral projection matrix is constructed as
\begin{equation}
P=\sum_{j\in\mathcal{J}_{\epsilon}}g_jg_j^{*},
\end{equation}
where $\mathcal{J}_{\epsilon}=\{j:|\lambda_j|>\epsilon\}$.
Each row of $P$ provides a direction-aware spectral representation of a vertex. The rows of $P$ are used as $\{z_1,\ldots, z_N\}$ for the $k$-means algorithm to obtain a $k$-way clustering of the digraph. When a real-valued implementation is required, each complex row $P_{i,:}$ can be represented as
\begin{equation}
\widetilde{P}_{i,:}=\left[\operatorname{Re}(P_{i,:}),\operatorname{Im}(P_{i,:})\right],
\end{equation}
which preserves the Euclidean distance between the original complex representations.

\begin{algorithm}[t]
\caption{Spectral Clustering for Digraphs with the Generalized Hermitian Matrix}
\label{alg:spectral_clustering}
\begin{algorithmic}[1]
\REQUIRE Digraph \(G=(V,E,A)\); $k\geq 2$; $(\mathcal{I}_{\mathrm{knw}},Y_{\mathrm{knw}})$; $\epsilon>0$; hyperparameters $\alpha \in [0,1]$ and $\beta \in (0,1]$
\STATE Construct the generalized Hermitian matrix:\\
    $W = \alpha(A + A^\top) + \beta(A - A^\top)\mathrm{i}$
\STATE Compute the eigenpairs $\{(\lambda_j, g_j)\}_{j\in \mathcal J_\epsilon}$ of $W$ with $\mathcal J_\epsilon= \{j: |\lambda_j|>\epsilon\}$.
\STATE $P\gets\displaystyle\sum_{j\in \mathcal J_\epsilon}g_jg_j^{*}$
\STATE Apply $k$-means on the matrix $\left[\operatorname{Re}(P),\operatorname{Im}(P)\right]$ subject to $(\mathcal{I}_{\mathrm{knw}},Y_{\mathrm{knw}})$ as in \eqref{k-means}
\RETURN A partition of $V$ based on the output of $k$-means
\end{algorithmic}
\end{algorithm}

\begin{algorithm}[t]
\caption{\textbf{\myproj}: A Semi-supervised Spectral Hierarchical Digraph Clustering Algorithm.}
\label{alg:semi_spechc}
\begin{algorithmic}[1]

\item[] \textbf{a) Input:}
A digraph $G=(V,E,A)$; $(\mathcal{I}_\mathrm{knw},Y_\mathrm{knw})$ if any; 
$K=(k_1,k_2,\ldots,k_L)$ satisfying $1<k_1<k_2<\cdots<k_L<N$, where $N=|V|$; hyperparameters $\alpha\in[0,1]$, $\beta\in(0,1]$, and $\epsilon>0$.

\item[] \textbf{b) Output:}
A hierarchical tree of levels $0,1,\ldots,L+1$, where level $0$ is the root $\{V\}$, level $L+1$ consists of the original vertices as leaves, and level $l$ contains a $k_l$-way clustering of $V$ for $l=1,\ldots,L$.
\item[] \textbf{c) Main Steps:}
\STATE \textbf{Initialization:}
$l\leftarrow L$, $V^{(L+1)}\leftarrow V$, $A^{(L+1)}\leftarrow A$; 
\WHILE{$l\geq 1$}
    \STATE Apply \textbf{Algorithm~\ref{alg:spectral_clustering}} to the current digraph with adjacency matrix $A^{(l+1)}$ and $k=k_l$, subject to the known pair $(\mathcal{I}_\mathrm{knw},Y_\mathrm{knw})$ resulting a $k_l$-way clustering $\mathcal{C}_l=\{C_1^{(l)},\ldots,C_{k_l}^{(l)}\}$ of $V$.
    \STATE Construct the coarse digraph
    $G^{(l)}=(V^{(l)},E^{(l)},A^{(l)})$ with $V^{(l)}=\{C^{(l)}_1,\ldots,C_{k_l}^{(l)}\}$ and adjacency matrix $A^{(l)}\in\mathbb{R}^{k_l\times k_l}$ given by \eqref{eq:cgc}.
    \STATE $l\leftarrow l-1$.
\ENDWHILE
\end{algorithmic}
\end{algorithm}

The complete digraph clustering procedure is summarized in
Algorithm~\ref{alg:spectral_clustering}. It serves as the basic
single-level clustering operation in the proposed hierarchical framework. In the following subsection, this procedure is recursively applied to successively coarsened digraphs to construct nested partitions.

\subsection{Spectral Hierarchical Digraph Clustering}
Algorithm~\ref{alg:spectral_clustering} produces a digraph partition at a prescribed granularity. To capture structures at multiple granularities, we recursively apply it to successively coarsened digraphs. The resulting semi-supervised spectral hierarchical digraph clustering method, termed \myproj, is summarized in Algorithm~\ref{alg:semi_spechc}.

Let $K=(k_1,\ldots,k_L)$ satisfy $1=k_0<k_1<\cdots<k_L<k_{L+1}=N$. The original vertices form level $L+1$ form $N$ clusters of singletons, level $l$ contains $k_l$ clusters from the level $l+1$, and level $0$ is the root representing a single cluster $\{V\}$ of all vertices. Starting from level $l=L+1$ with input the initial digraph $G^{(L+1)}=G$, Algorithm~\ref{alg:spectral_clustering} is recursively applied to the current digraph $G^{(l+1)}=(V^{(l+1)}, E^{(l+1)},A^{(l+1)})$ with $k=k_l$, yielding a $k_l$-way clustering  $\mathcal{C}_l$ of $V$. 
Note that the known labels in the pair $(\mathcal{I}_\mathrm{knw},Y_\mathrm{knw})$ with $Y_{\mathrm{knw}}(\mathcal{I}_\mathrm{knw})\subseteq\{1,\ldots,k'\}$ for some $k'\le k_1$ are imposed as fixed cluster-assignment constraints during the $k$-means clustering. Specifically, let $\mathcal{C}_l=\{C_1^{(l)},\ldots,C_{k_l}^{(l)}\}$ be the $k_l$-way clustering of $V=\{v_1,\ldots, v_N\}$ obtained at level $l$. Then, it must satisfy $C_j^{(l)}\supseteq Y_{\mathrm{knw}}^{-1}(j)$ for $j=1,\ldots, k'$. When $\mathcal{I}_\mathrm{knw}=\emptyset$, \myproj~reduces to the unsupervised hierarchical clustering. 
The partition $\mathcal{C}_l$ is then contracted to construct the next coarser digraph $G^{(l)}=(V^{(l)}, E^{(l)}, A^{(l)})$ with $V^{(l)}=\mathcal C_l$. Each cluster becomes a coarse vertex, and the directed weight from $C_i^{(l)}$ to $C_j^{(l)}$ is defined as
\begin{equation}
\label{eq:cgc}
A^{(l)}(i,j)=\frac{1}{\operatorname{vol}(G)}\sum_{C_u^{(l+1)}\subseteq C_i^{(l)}}\sum_{C_v^{(l+1)}\subseteq C_j^{(l)}}A^{(l+1)}(u,v),
\end{equation}
for $1\leq i,j\leq k_l$. Since the aggregation is performed separately for each ordered pair, $A^{(l)}(i,j)$ and $A^{(l)}(j,i)$ may differ, thereby preserving asymmetric inter-cluster interactions. 

The resulting clusterings satisfy
\begin{equation}\label{tree-C}
\mathcal{C}_{L+1}\preceq\mathcal{C}_L\preceq\mathcal{C}_{L-1}\preceq\cdots\preceq\mathcal{C}_1\preceq \mathcal{C}_0,
\end{equation}
where $\preceq$ denotes partition containment. These containment relations by the above \emph{bottom-up} approach define a hierarchical tree (filtration), with the original vertices as leaves at the bottom and $\{V\}$ as the root at the top. See Fig.~\ref{fig:toy_example} (top row, left to right). 

\begin{figure*}[htb]
\setlength{\abovecaptionskip}{0.cm}
\setlength{\belowcaptionskip}{-0.cm}
\centering
\includegraphics[width=0.75\textwidth]{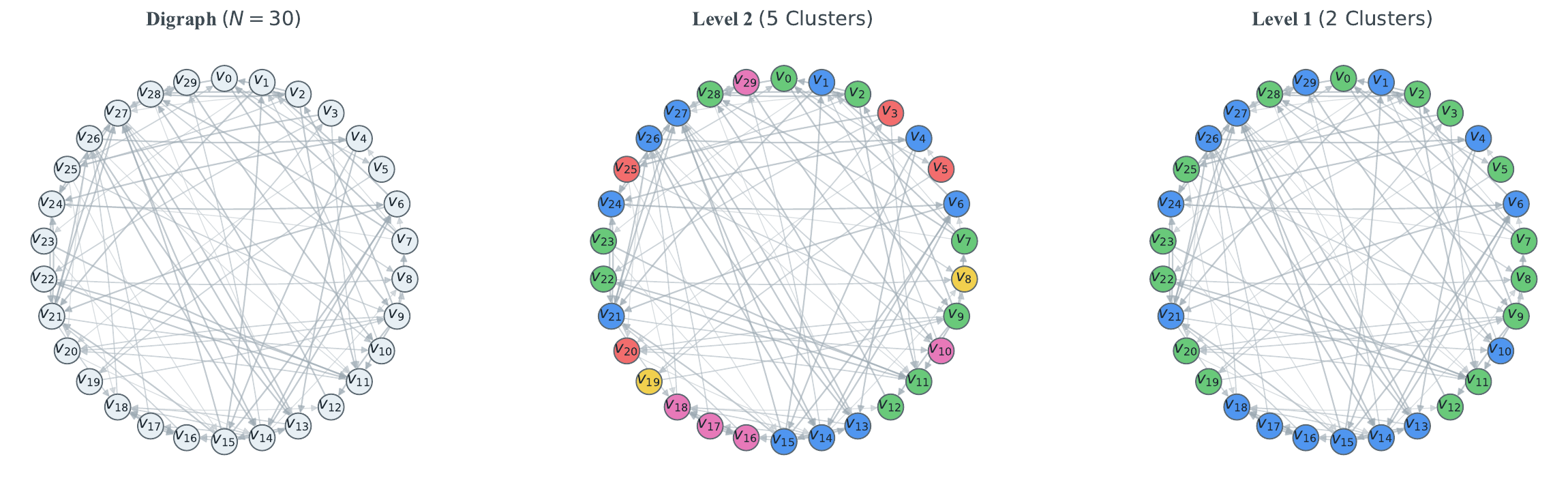}\\
\includegraphics[width=0.25\textwidth]{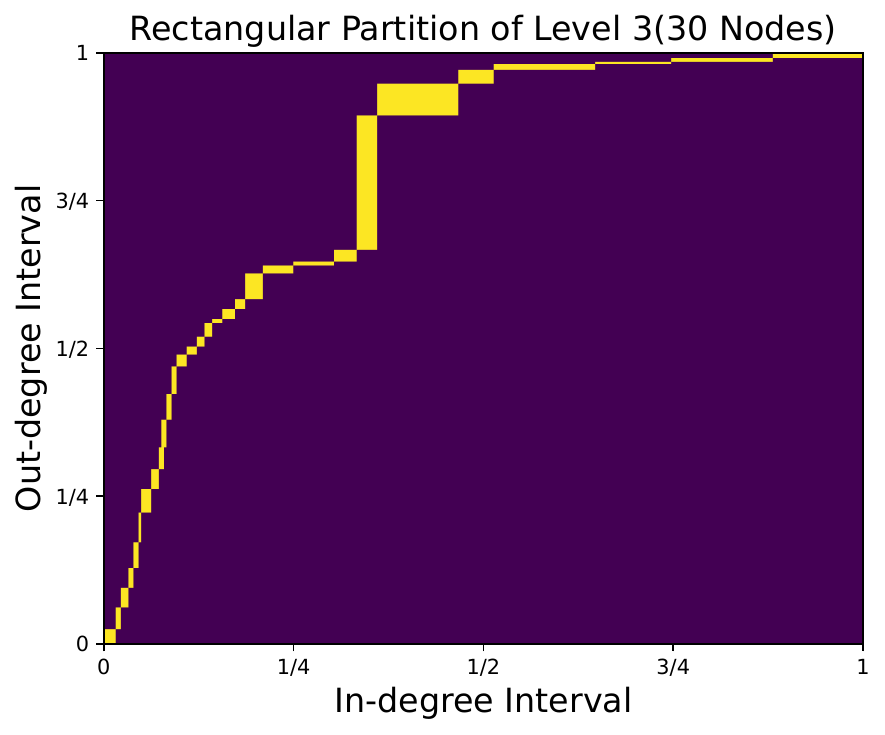}
\includegraphics[width=0.25\textwidth]{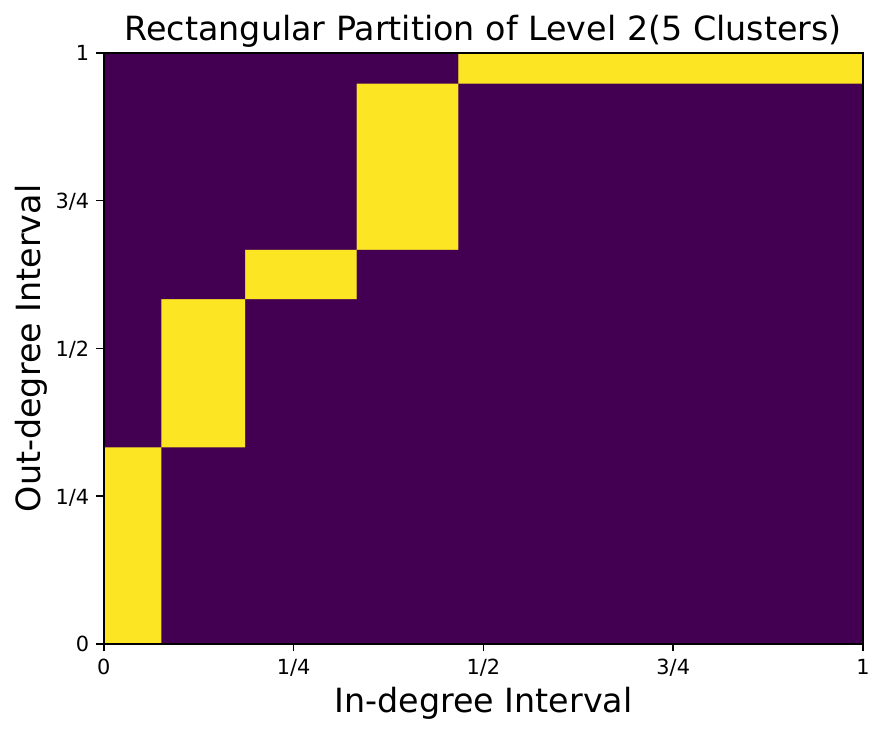}
\includegraphics[width=0.25\textwidth]{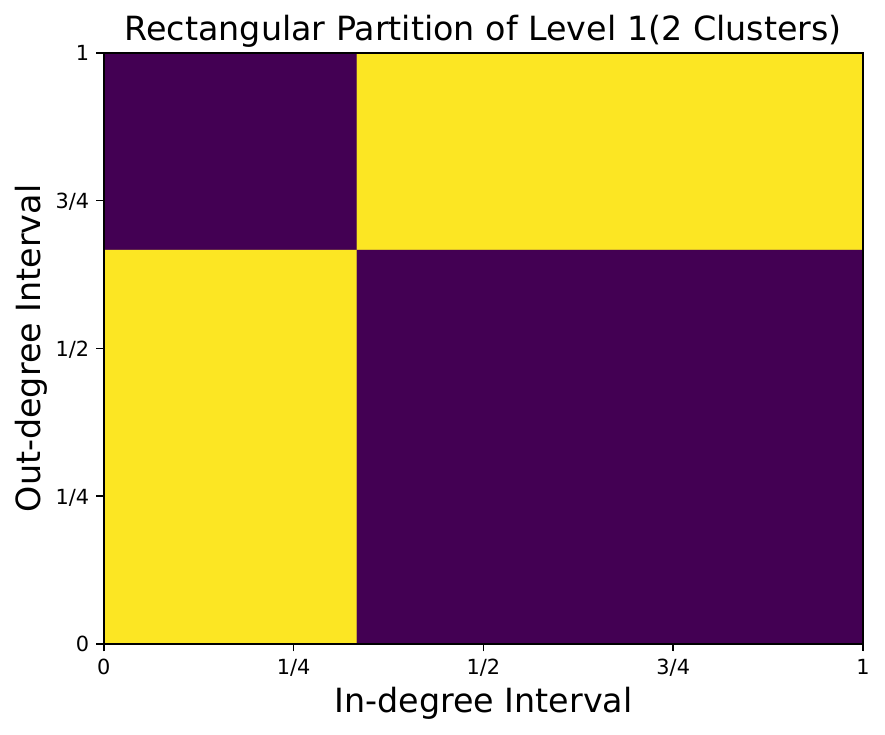}
\caption{The top presents the original digraph and the hierarchical clustering results obtained with our algorithm using $\alpha(A+A^T)+\beta(A-A^T)\mathrm{i}$. The bottom displays the corresponding rectangular partition of $I^2 = [0, 1] \times [0, 1]$, derived from the hierarchical structure of the trees above.}\label{fig:toy_example}
\end{figure*}
\subsection{Hierarchy-Induced Interval Partitions}\label{sec:interval-partition}
We now associate the hierarchy in \eqref{tree-C} to a pair of nested interval partitions in a \emph{top-down} manner. Specifically, along the in-degree and out-degree coordinates, respectively, from the root level $0$ to the leaf level $L+1$, we associate each cluster $C^{(l)}_j$ at each level with a block $R^{(l)}_j=I_{j,\mathrm{in}}^{(l)}\times I_{j,\mathrm{out}}^{(l)}$ of sub-intervals in $[0,1]$. Such a procedure also provides a sequence of nested knot sequences for the definition of splines in Section~\ref{sec:splines}.

More precisely, at level $0$, the root cluster $\{V\}$ is assigned the full  unit square $[0,1]\times [0,1]$, that is, $I_{\mathrm{in}}^{(0)}:=[a_{\mathrm{in}}^{(0)},b_{\mathrm{in}}^{(0)}]=[0,1]= [a_{\mathrm{out}}^{(0)},b_{\mathrm{out}}^{(0)}]=:I_{\mathrm{out}}^{(0)}$. Now, from the top level going down along the tree in \eqref{tree-C} for $l=1,\ldots,L+1$, suppose that a parent cluster $C^{(l-1)}$ at level $l-1$ is associated with  $[a_{\mathrm{in}}^{(l-1)},b_{\mathrm{in}}^{(l-1)}]\times [a_{\mathrm{out}}^{(l-1)},b_{\mathrm{out}}^{(l-1)}]$ at level $l-1$ and it is formed from the ordered children clusters $C_1^{(l)},\ldots,C_k^{(l)}$ at level $l$. That is, $C^{(l-1)}=\cup_j C^{(l)}_j$. Let $G^{(l)}=(V^{(l)},E^{(l)},A^{(l)})$ be the digraph obtained at level $l$. Based on the adjacency matrix $A^{(l)}=[a_{uv}^{(l)}]$, the aggregate in-degree and out-degree for $u=C_j^{(l)}$ are computed as 
\[
\begin{aligned}
d_{\mathrm{in}}\left(C_j^{(l)}\right)=\sum_{v\in V^{(l)}}a_{vu}^{(l)},\,\, d_{\mathrm{out}}\left(C_j^{(l)}\right)=\sum_{v\in V^{(l)}}a_{uv}^{(l)}.
\end{aligned}
\]
Then, the directional proportion (relative weight) of $C_j^{(l)}$ among all the ordered children clusters can be defined as $\rho_{j,\iota}^{(l)}=d_{\iota}\left(C_j^{(l)}\right) /\sum_{i=1}^{j}d_{\iota}\left(C_i^{(l)}\right)$ for $\iota\in\{\mathrm{in},\mathrm{out}\}$.
The parent intervals are then subdivided according to these proportions (weights). Specifically,  setting $t_{0,\iota}^{(l)}=a^{(l-1)}_{\iota}$, we define
\[
t_{j,\iota}^{(l)}=a^{(l-1)}_{\iota}+\left(b^{(l-1)}_{\iota}-a^{(l-1)}_{\iota}\right)\sum_{i=1}^{j}\rho_{i,\iota}^{(l)}, \, \iota\in\{\mathrm{in},\mathrm{out}\},
\]
for $j=1,\ldots,k$. The cluster $C_j^{(l)}$ is then associated with a block  $I_{j,\mathrm{in}}^{(l)}\times I_{j,\mathrm{out}}^{(l)}$, where
$I_{j,\iota}^{(l)}=[t_{j-1,\iota}^{(l)},t_{j,\iota}^{(l)}]$ for $\iota\in\{\mathrm{in},\mathrm{out}\}$.
Applying this subdivision procedure to all parent clusters at level $(l-1)$ produces the set $\mathcal{R}_l=\{R_j^{(l)}={I}^{(l)}_{j,\mathrm{in}}\times {I}^{(l)}_{j,\mathrm{out}}: j=1,\ldots,|V^{(l)}|\}$ of blocks at level $l$. 

Note that $[0,1]=\cup_{j=1}^{|V^{(l)}|}I^{(l)}_{j,\iota}$ for any $\iota\in\{\mathrm{in, out}\}$ and any $l=0,\ldots,L+1$. The resulting block sequence satisfies
\begin{equation}\label{tree-B}
\mathcal{R}_{0}\succeq\mathcal{R}_1\succeq\mathcal{R}_{2}\succeq\cdots\succeq\mathcal{R}_L\succeq \mathcal{R}_{L+1},
\end{equation}
where $\succeq$ denotes partition inclusion. 
See Fig.~\ref{fig:toy_example} (bottom row, right to left). 

\subsection{Illustration and Synthetic Examples}
To illustrate the hierarchical clustering process in Section~\ref{sec:digraph_clustering}, we use a synthetic digraph generated by the directed stochastic block model (DSBM). 
The DSBM was introduced by Cucuringu et al. in \cite{cucuringu2020hermitian} (see also \textit{Supplementary Material Part~A.5}), where a directed graph $G=(V,E,A)$ drawn from the model $\mathcal G(k,n,p,q,F)$ has $k$ clusters and $|V|=N=k\cdot n$ (i.e., each cluster has $n$ vertices). 
The parameter $p$ controls the probability that there is an edge between two vertices within the same cluster, the parameter $q \in [0, 1]$ controls the probability that there is an edge between two vertices belonging to two different clusters, and $F \in [0, 1]^{k\times k}$ is a matrix where $F_{j,j'}$ denotes the probability of an edge from cluster $j$ to cluster $j'$. Necessarily, $F_{j,j'}+F_{j',j}=1$, and $F_{j,j}=1/2$. The pair $(p,q)$ affects the homophily ratio of the directed graph (see \eqref{def:homophiy} and Section~\ref{sec:DSBM}). 

As an example, we generate a graph from the model with $p=q=0.5$ and $N=30$, using $k=2$ and $n=15$. The cyclic pattern matrix is $F=
\bigl[\begin{smallmatrix}
\frac{1}{2} & 1-\eta\\
\eta & \frac{1}{2}
\end{smallmatrix}\bigr]$,
where $0\le\eta\le\frac12$ controls the directional randomness between the two clusters. Smaller $\eta$ indicates a stronger directional preference, while $\eta=\frac12$ corresponds to no directional bias. We set $\eta=0$, giving $F=\bigl[\begin{smallmatrix}1/2&1\\0&1/2\end{smallmatrix}\bigr]$. We then apply \myproj~in the unsupervised setting with $\mathcal{I}_{\mathrm{knw}}=\emptyset$ and set $K=(k_1,k_2)=(2,5)$. Thus, Levels~1 and~2 contain $2$ and $5$ clusters, respectively, while Level~3 contains the $30$ original vertices. 

As shown in Fig.~\ref{fig:toy_example}, the top row (from left to right) presents the original digraph at level 3 and the clustering results at levels $2$ and $1$, respectively.  Vertices with the same color belong to the same cluster.
The $N=30$ vertices at level 3 are aggregated to 5 clusters at level 2, and they are further grouped into two larger clusters at level $1$, demonstrating the nested structure $V=\mathcal C_3\preceq \mathcal C_2\preceq\mathcal C_1\preceq \mathcal C_0=\{V\}$ produced by the \myproj~alogirithm (as a bottom-up procedure). Note that the root level $\mathcal C_0$ is omitted as it is the single cluster $\{V\}$. The bottom row of Fig.~\ref{fig:toy_example} (from right to left) shows the induced hierarchy block partitions $\mathcal R_0\succeq \mathcal R_1\succeq\mathcal R_2\succeq \mathcal R_3$ (by a top-down procedure). At each level, the in-degree and out-degree intervals assigned to the same cluster are paired to form a rectangular cell. The level-$1$ partition $\mathcal R_1$ (bottom right), therefore, contains two rectangle blocks, while the level-$2$ partition $\mathcal R_2$ (bottom middle) contains five. At the finest level $\mathcal R_3$ (bottom left), each original vertex is associated with an individual block cell. The successive subdivisions illustrate how the clustering hierarchy induces increasingly refined directional partitions and nested knot sequences for subsequent spline quasi-interpolation.


\section{Experiments on Digraph Clustering}
\label{sec:experiments}
We evaluate our method against representative spectral clustering approaches on both synthetic and real-world directed graphs. The experiments further include parameter sensitivity analysis and a comparison of different normalization strategies. Section~\ref{sec:DSBM} reports results on synthetic DSBM datasets, examining how clustering performance changes with key model parameters and revealing the structural conditions under which each method performs best. Section~\ref{sec:real_world} presents evaluations on real-world networks under different training settings, comparing our method with the baselines across diverse graph types. Additional experimental details and results are provided in 
\textit{Supplementary Material Parts~A and B}.

\subsection{Experimental Setup} \label{sec:ExperimentalSetup}

\noindent\textbf{Datasets. }The experiments are conducted on both synthetic and real-world datasets. Synthetic digraphs are generated using the DSBM with the five parameters: number of nodes $N$, number of clusters $k$, intra-cluster link probability $p$, inter-cluster link probability $q$, and the direction probability matrix $F$, allowing us to systematically control their structural properties. The entries of $F$ are restricted to $1/2$, $\eta$, and $1-\eta$, corresponding to random, preferred, and reversed edge directions, respectively. Unless otherwise stated, we set $N=5000$ and $k=5$ by default. In addition, we evaluate our model on seven real-world directed graph datasets with diverse scales, structural properties, and homophily levels (see \eqref{def:homophiy}).
These datasets span diverse application domains, including the Cora citation network~\cite{sen2008collective}, the Squirrel Wikipedia webpage network~\cite{rozemberczki2021multi}, the Telegram influence network~\cite{zhang2021magnet,he2022digrac}, the Cornell, Wisconsin, and Texas WebKB networks~\cite{pei2020geom}, and the political Blog network~\cite{adamic2005political}. Since Telegram and Blog are unlabeled, their class counts and homophily ratios are reported as N/A. All datasets are publicly available and have been widely used in directed graph learning, supporting fair and reproducible evaluation. Due to space limitations, the experimental results for Squirrel, Cornell, and Texas are provided in \textit{Supplementary Material Part~B.2}, while detailed statistics and descriptions of all datasets are summarized in \textit{Supplementary Material Part~A.2}.
For a  directed graph $G=(V,E,A)$ with each node $v$ with a label $y_v\in\{1,\ldots,k\}$, its \emph{homophily ratio} \cite{huang2024flow2gnn} is defined by
\begin{equation}\label{def:homophiy}
\mathcal{H}_G=\frac{1}{|V|}\sum_{v\in V} h_v,  
\end{equation}
with the node homophily ratio at $v$ given by
\[
h_v=\frac{1}{2}\!\left(\!\frac{|\{u\in\mathop{\mathcal{N}}\limits ^{\leftarrow}(v)\!:\!{y}_v\!=\!{y}_u\}|\!}{|\mathop{\mathcal{N}}\limits ^{\leftarrow}(v)|\!}\!+\!\frac{\!|\{u\in\mathop{\mathcal{N}}\limits ^{\rightarrow}(v)\!:\!{y}_v\!=\!{y}_u\}|}{\!|\mathop{\mathcal{N}}\limits ^{\rightarrow}(v)|}\!\right)\!,
\]
where $\mathop{\mathcal{N}}\limits ^{\leftarrow}(v)$ denotes the set of nodes that have edges pointing into node $v\in {V}$, and $\mathop{\mathcal{N}}\limits ^{\rightarrow}(v)$ is the set of nodes that have edges pointing out of node $v$. 

\noindent\textbf{Competitors.} 
We benchmark our approach against five state-of-the-art spectral clustering methods designed for directed graphs: Bi-Sym \cite{satuluri2011symmetrizations}, DD-Sym \cite{satuluri2011symmetrizations}, DI-SIM \cite{rohe2016co}, Herm \cite{cucuringu2020hermitian}, and Skew \cite{hayashi2022skew}. These unsupervised baselines primarily rely on matrix symmetrization or spectral decomposition to learn node embeddings. All methods are strictly evaluated under their original hyperparameter settings across synthetic DSBM graphs and the seven real-world directed networks. Additional details are provided in \textit{Supplementary Material Parts A.3}. 

\noindent\textbf{Evaluation Metrics.}
We evaluate algorithm performance using three metrics: Adjusted Rand Index (ARI) \cite{gates2017impact}, Modularity \cite{malliaros2013clustering}, and F-measure \cite{satuluri2011symmetrizations}. On synthetic graphs generated from the DSBM, we use ARI, as the ground-truth communities are known. On real-world datasets, modularity and F-measure metrics are adopted for adapting both unsupervised and semi-supervised settings. Modularity does not require labels, whereas ARI and F-measure do. All metrics reflect how closely the recovered clusters match the ground truth: values near 1 indicate nearly perfect recovery, while values near 0 suggest nearly random partitioning. 

\noindent\textbf{Implementation Details. }
All experiments are conducted on a workstation equipped with an NVIDIA GeForce RTX 2080 Ti GPU and 24 GB of memory. The proposed method is implemented in Python, with the eigendecomposition and clustering procedures performed using standard numerical and machine-learning libraries. Unless otherwise specified, the parameters $\alpha$, $\beta$, and $\epsilon$ are selected through grid search on the validation data. For semi-supervised experiments, the labeled vertices are randomly sampled from each class to ensure that all classes are represented. Each experiment is repeated over multiple random splits, and the mean performance is reported.

\subsection{Results for the DSBM} \label{sec:DSBM}


\noindent\textbf{Results and Analysis.} 
We conduct experiments on graphs generated randomly from the DSBM with different values of $n$, $p$, $q$, and matrix $F$, as shown in Fig. \ref{fig:p_q}. In these graphs, the gap between $p$ and $q$ ranges from twice to ten times. The generated directed graphs are either homophilic or heterophilic. Under both scenarios, we observe how other DSBM parameters affect the results and obtain several key findings. All reported results are averaged over 10 independently generated graphs for each fixed parameter set. Additional results are reported in \textit{Supplementary Material Part~B.1}.

\noindent\textbf{Observation 1: Clear homophilic or heterophilic structures facilitate accurate clustering.}
Fig.~\ref{fig:p_q} examines the performance of \myproj~under different combinations of the DSBM parameters $p$ and $q$. When $p\gg q$, edges are mainly formed within clusters, producing a pronounced homophilic structure. Conversely, $q\gg p$ yields a clear heterophilic structure dominated by inter-cluster connections. In both cases, the ARI increases as the difference between $p$ and $q$ becomes larger. When $p$ and $q$ are close, the cluster boundaries become less distinguishable and the performance decreases. Therefore, \myproj~performs most reliably when the generated digraph exhibits a clear homophilic or heterophilic pattern.

\noindent\textbf{Observation 2: Directional randomness mainly affects graphs with ambiguous structural patterns.}
Fig.~\ref{fig:p_q} also shows the performance of \myproj~as the directionality parameter $\eta$  increases. A larger $\eta$ introduces more randomness into edge directions and weakens the directional distinction
between clusters. In the top panel, strongly homophilic graphs maintain ARI values close to $1$ over the entire range of $\eta$, whereas the ARI values for weakly or moderately homophilic graphs deteriorate more rapidly. A similar trend is
observed in the bottom panel: clearly, heterophilic graphs remain stable, while graphs with weaker inter-cluster separation become increasingly sensitive to directional randomness. These results indicate that structural separation can compensate for noisy edge directions, whereas ambiguous structures rely more strongly on consistent directional information.

\begin{figure}[t]
\setlength{\abovecaptionskip}{0.cm}
\setlength{\belowcaptionskip}{-0.cm}
\centering
\includegraphics[width=0.40\textwidth]{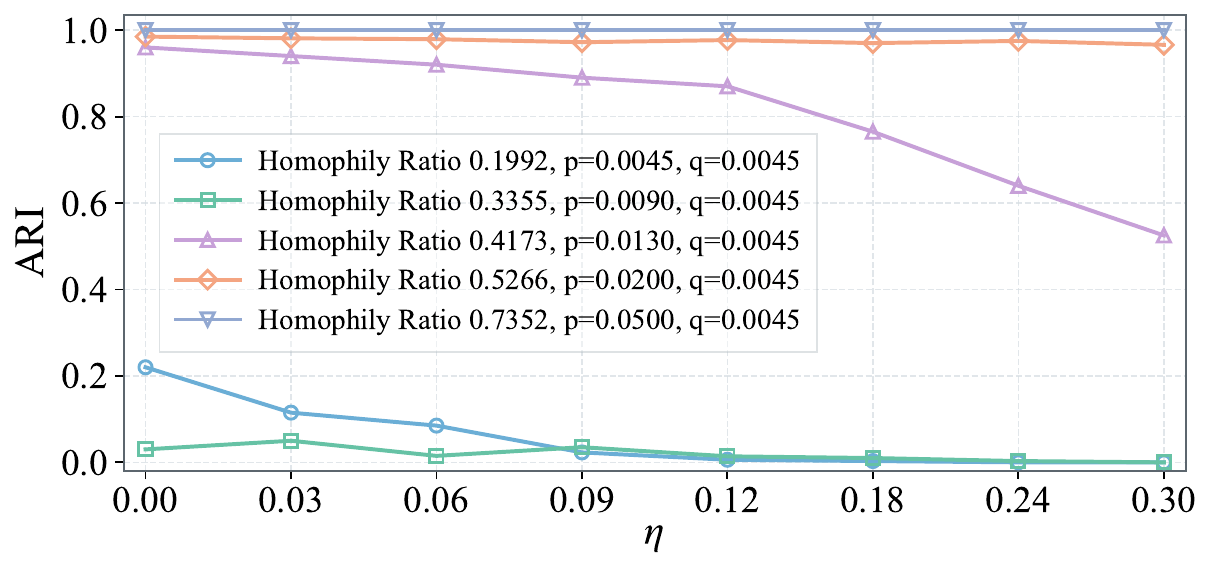}
\includegraphics[width=0.40\textwidth]{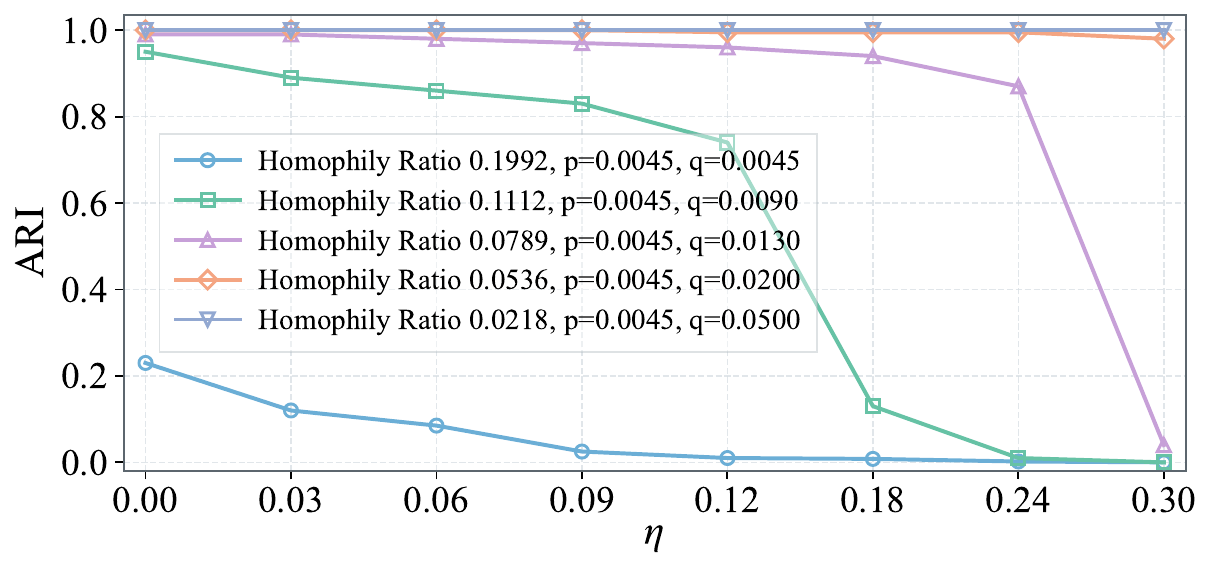}
\caption{Sensitivity of clustering performance to three key parameters $p,q,\eta$ on DSBM graphs. The top panel varies the within-cluster probability $p$ while fixing $q=0.0045$, producing graphs with increasing homophily ratios. The bottom panel varies the between-cluster probability $q$ while fixing $p=0.0045$, producing graphs with decreasing homophily ratios.}\label{fig:p_q}
\end{figure}


\begin{figure*}[t!]
\setlength{\abovecaptionskip}{0.cm}
\setlength{\belowcaptionskip}{-0.cm}
\begin{center}
\subfigure[$\eta=0$, $p=0.0045$, $q=0.0045$]{\includegraphics[width=0.28\textwidth]{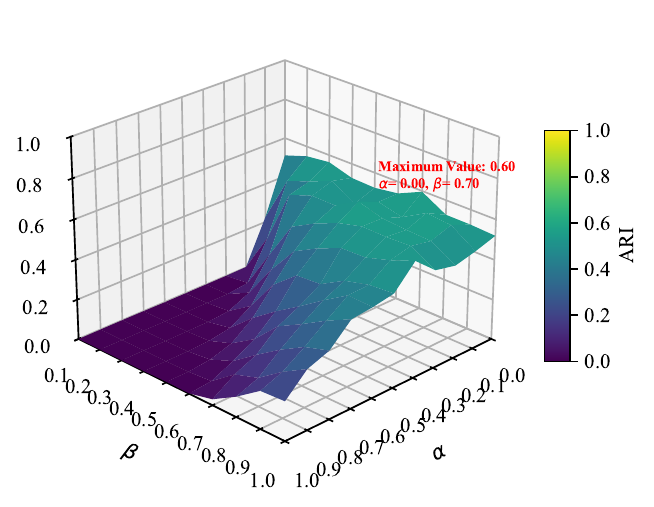}}
\subfigure[$\eta=0$, $p=0.0045$, $q=0.05$]{\includegraphics[width=0.28\textwidth]{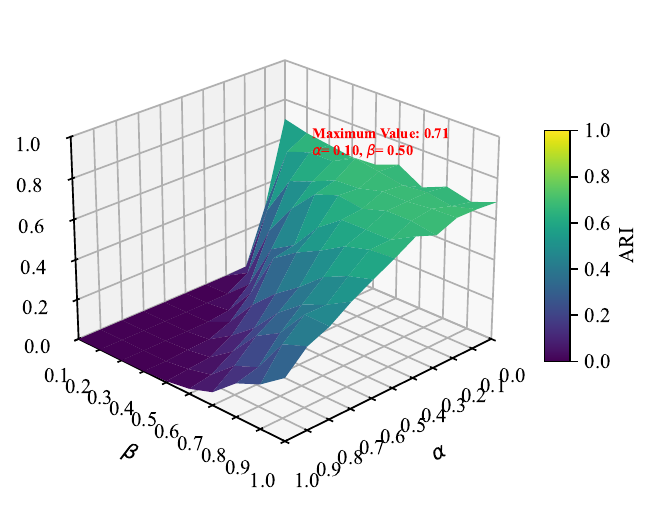}}
\subfigure[$\eta=0$, $p=0.05$, $q=0.0045$]{\includegraphics[width=0.28\textwidth]{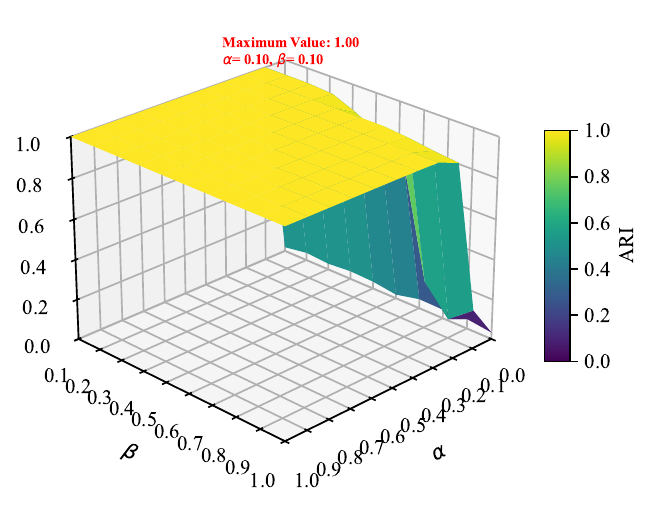}}
\caption{Investigating the effect of $\alpha$ and $\beta$: Model parameter sensitivity on different directed graph structures.}
\label{fig:alpha_beta}
\end{center}
\end{figure*}
\begin{table*}[htpb!]
\caption{Clustering algorithms on the Wisconsin dataset: a comparative analysis of modularity ($\mathcal{M}$), and F-measure ($\mathcal{F}$). The best-performing model is highlighted in \textbf{bold}, and the second-best is marked with an \underline{underline}.}\label{Table:Wisconsin}
\vspace*{1mm}
\centering
\scalebox{0.90}[0.90]{\begin{tabular}{llcccccccccc}
\toprule
$\mathcal{M}$& \textbf{Trains (\%)}  & \textbf{0 (USL)}  & \textbf{10}  & \textbf{20}  & \textbf{30}  & \textbf{40}  & \textbf{50}  & \textbf{60}  & \textbf{70}  & \textbf{80}  & \textbf{90}  \\ \midrule
& \textbf{Level 2 (10)} & 0.2262 & 0.2505 & 0.2484 & 0.1880 & 0.1549 & 0.0923 & 0.0933 & 0.0327 & 0.0593 & 0.0330  \\
\multirow{-2}{*}{\textbf{Bi-Sym}} & \textbf{Level 1 (5)} & 0.1524 & 0.0971 & 0.0848 & \textbf{0.3868} & 0.1712 & \textbf{0.2787} & 0.1995 & 0.0940 & 0.2439 & 0.0071  \\
& \textbf{Level 2 (10)} & 0.5171 & \textbf{0.5031} & \underline{0.3864} & 0.1777 & \textbf{0.2956} & 0.1919 & 0.0829 & 0.0299 & 0.0415 & 0.0048  \\
\multirow{-2}{*}{\textbf{DD-Sym}} & \textbf{Level 1 (5)} & \textbf{0.5305} & 0.2839 & 0.3454 & 0.2318 & 0.0751 & 0.0999 & 0.2017 & 0.1596 & 0.2086 & \textbf{0.1823}  \\
& \textbf{Level 2 (10)} & \underline{0.5192} & 0.3690 & \textbf{0.5065} & \underline{0.3297} & 0.2122 & 0.1023 & 0.1046 & 0.0007 & 0.0012 & 0.0421  \\
\multirow{-2}{*}{\textbf{DI-SIM}} & \textbf{Level 1 (5)} & 0.1281 & 0.3061 & 0.1822 & 0.3266 & 0.1528 & 0.1686 & \underline{0.2333} & \underline{0.1802} & \textbf{0.3917} & 0.0638  \\
& \textbf{Level 2 (10)} & 0.1968 & 0.1769 & 0.1675 & 0.1613 & 0.1050 & 0.0622 & 0.0791 & 0.0302 & 0.0554 & 0.0340  \\
\multirow{-2}{*}{\textbf{Herm}} & \textbf{Level 1 (5)} & 0.0635 & 0.0205 & 0.1077 & 0.1328 & 0.2164 & 0.1179 & \textbf{0.2759} & 0.1696 & 0.1617 & \underline{0.1540}  \\
& \textbf{Level 2 (10)} & 0.1968 & 0.1769 & 0.1675 & 0.1613 & 0.1050 & 0.0622 & 0.0791 & 0.0302 & 0.0554 & 0.0340  \\
\multirow{-2}{*}{\textbf{Skew}} & \textbf{Level 1 (5)} & 0.0116 & 0.0303 & 0.1265 & 0.1124 & 0.1933 & \underline{0.2340} & 0.1506 & 0.0625 & 0.1497 & 0.0250  \\
\midrule
& \textbf{Level 2 (10)} & 0.2478 & \underline{0.4298} & 0.2378 & 0.1794 & 0.0892 & 0.2086 & 0.0993 & 0.0565 & 0.0075 & 0.0310  \\
\multirow{-2}{*}{\textbf{\myproj}} & \textbf{Level 1 (5)} & 0.4169 & 0.1416 & 0.1954 & 0.1938 & \underline{0.2403} & 0.1024 & 0.1101 & \textbf{0.2356} & \underline{0.2589} & 0.0873  \\
\toprule
$\mathcal{F}$& \textbf{Trains (\%)}  & \textbf{0 (USL)}  & \textbf{10}  & \textbf{20}  & \textbf{30}  & \textbf{40}  & \textbf{50}  & \textbf{60}  & \textbf{70}  & \textbf{80}  & \textbf{90}  \\ \midrule
& \textbf{Level 2 (10)} & 0.1718 & 0.1847 & 0.2339 & 0.3744 & 0.3522 & 0.4259 & 0.6316 & 0.6682 & 0.7444 & 0.8856  \\
\multirow{-2}{*}{\textbf{Bi-Sym}} & \textbf{Level 1 (5)} & 0.2698 & \textbf{0.5105} & 0.1763 & 0.2946 & 0.5118 & 0.6117 & 0.7082 & 0.7673 & 0.8487 & 0.9021  \\
& \textbf{Level 2 (10)} & 0.0571 & 0.3820 & 0.1490 & 0.3276 & 0.4825 & 0.4485 & 0.5995 & 0.6905 & 0.7514 & 0.8688  \\
\multirow{-2}{*}{\textbf{DD-Sym}} & \textbf{Level 1 (5)} & 0.3835 & 0.1299 & 0.2125 & 0.3677 & 0.5833 & 0.5866 & 0.6478 & 0.7470 & 0.8169 & 0.8989  \\
& \textbf{Level 2 (10)} & 0.0650 & 0.4695 & 0.4715 & 0.2089 & 0.5614 & 0.4042 & 0.5201 & 0.6869 & 0.7465 & 0.9269  \\
\multirow{-2}{*}{\textbf{DI-SIM}} & \textbf{Level 1 (5)} & 0.1267 & \underline{0.4978} & \underline{0.5509} & \underline{0.5017} & 0.6362 & 0.5877 & \underline{0.7396} & 0.6890 & \textbf{0.8859} & \textbf{0.9477}  \\
& \textbf{Level 2 (10)} & 0.1336 & 0.2222 & 0.1930 & 0.2501 & 0.3164 & 0.5370 & 0.6029 & 0.6283 & 0.7458 & 0.8889  \\
\multirow{-2}{*}{\textbf{Herm}} & \textbf{Level 1 (5)} & 0.3696 & 0.3101 & 0.2783 & 0.4236 & \underline{0.6599} & \underline{0.6194} & 0.6584 & \textbf{0.8237} & 0.8311 & 0.9016  \\
& \textbf{Level 2 (10)} & 0.1336 & 0.2222 & 0.1930 & 0.2501 & 0.3164 & 0.5370 & 0.6029 & 0.6283 & 0.7458 & 0.8889  \\
\multirow{-2}{*}{\textbf{Skew}} & \textbf{Level 1 (5)} & \underline{0.3904} & 0.3193 & 0.3540 & 0.4255 & 0.4989 & 0.5688 & 0.7009 & \underline{0.8116} & \underline{0.8535} & 0.9111  \\
\midrule
& \textbf{Level 2 (10)} & 0.0493 & 0.3010 & 0.1377 & 0.3245 & 0.4856 & 0.4323 & 0.6300 & 0.7305 & 0.7672 & 0.8982  \\
\multirow{-2}{*}{\textbf{\myproj}} & \textbf{Level 1 (5)} & \textbf{0.5836} & 0.4923 & \textbf{0.5980} & \textbf{0.6198} & \textbf{0.6785} & \textbf{0.7141} & \textbf{0.7940} & 0.7681 & 0.8408 & \underline{0.9358}  \\
\bottomrule
\end{tabular}}
\end{table*}

\noindent\textbf{Observation 3: The appropriate balance between $\alpha$ and $\beta$ depends on graph structure.}
Fig.~\ref{fig:alpha_beta} studies the sensitivity of the unnormalized \myproj~to $\alpha$ and $\beta$, which control the symmetric and asymmetric components of the generalized Hermitian matrix, respectively. When $p=q$, the model generally benefits from a stronger asymmetric component, although the overall performance remains limited by the weak cluster separation. A similar tendency is observed when $q\gg p$, where directional inter-cluster interactions provide the main clustering information. In contrast, when $p\gg q$, larger values of $\alpha$ or a more balanced combination of $\alpha$ and $\beta$ are preferable because the within-cluster connectivity itself becomes highly informative. Thus, no single parameter combination is optimal for all digraphs; $\alpha$ and $\beta$ should be
balanced according to the relative importance of connectivity and directionality.

\subsection{Results for Real-World Data} \label{sec:real_world}
We evaluate the competing methods on both labeled and unlabeled real-world digraphs. For labeled datasets, Modularity ($\mathcal{M}$) and F-measure ($\mathcal{F}$) are reported for labeled-node ratios ranging from $0\%$ to $90\%$, where $0\%$ corresponds to the unsupervised setting (ULS). Results are evaluated at two hierarchical resolutions: Level~1 contains five coarse clusters, while Level~2 contains ten finer clusters. The two metrics provide complementary views: $\mathcal{M}$ measures structural cohesion, whereas $\mathcal{F}$ evaluates agreement with the ground-truth classes.

Table~\ref{Table:Wisconsin} reports the results on Wisconsin. At Level~1, \myproj~achieves the best $\mathcal{F}$ at labeled-node ratios of $0\%$, $20\%$, $30\%$, $40\%$, $50\%$, and $60\%$, and obtains the second-best results at $10\%$ and $90\%$. Its score increases from $0.5836$ without labels to $0.9358$ at the $90\%$ setting, showing that the inherited label constraints effectively improve the class consistency of the coarse partition. At Level~2, \myproj~becomes increasingly competitive as more labels are provided, attaining the best $\mathcal{F}$ at $70\%$ and $80\%$
and the second-best result at $60\%$ and $90\%$. This indicates that stronger supervision is particularly useful for distinguishing finer substructures. The modularity results exhibit a less regular trend because structurally cohesive communities do not necessarily coincide with the ground-truth classes. Nevertheless, \myproj~achieves the highest $\mathcal{M}$ at the $50\%$ and $70\%$ settings of Level~2 and at the $70\%$ setting of Level~1, while remaining competitive in several other cases. The results therefore
show that the two hierarchy levels capture complementary properties: Level~1 is generally more consistent with the class partition, whereas Level~2 provides a finer description of the graph structure.

Table~\ref{Table:Unlabeled} reports ARI ($\mathcal{A}$) on the unlabeled Telegram and Blog networks. Since no reference labels are available, ARI is used here to measure the consistency of the partitions across repeated unsupervised runs rather than their agreement with ground-truth classes. \myproj~achieves the highest stability on Telegram with an ARI of $0.97$, followed by DD-Sym at $0.95$. On Blog, \myproj~and DD-Sym jointly obtain the best score of $0.88$. These results indicate that \myproj~produces highly repeatable partitions on both unlabeled networks.

Overall, \myproj~shows strong performance across unsupervised and
semi-supervised settings, particularly for class-aligned coarse partitions, high-supervision fine partitions, and clustering stability on unlabeled digraphs. Due to space limitations, the results on Cora, Squirrel, Cornell, and Texas are reported in \textit{Supplementary Material Part~B.2}.

\begin{table}[tbp]
\caption{Clustering algorithms on the two unlabeled datasets: a comparative analysis of ARI ($\mathcal{A}$). The reported ARI measures clustering stability across repeated unsupervised runs. The best results are shown in \textbf{bold}, and the second-best are \underline{underlined}.}\label{Table:Unlabeled}
\vspace*{1mm}
\centering
\scalebox{0.89}[0.89]{\begin{tabular}{c|ccccc|c}
\toprule
\textbf{Datasets} & \textbf{Bi-Sym} & \textbf{DD-Sym} & \textbf{DI-SIM} & \textbf{Herm} & \textbf{Skew} & \textbf{\myproj} \\ \midrule
telegram & 0.91 & \underline{0.95} & 0.47 & 0.89 & 0.89 & \textbf{0.97} \\
blog & 0.82 & \textbf{0.88} & 0.64 & \underline{0.84} & \underline{0.84} & \textbf{0.88} \\\bottomrule
\end{tabular}}
\end{table}

\section{Stability of spline quasi-interpolation with doubling measures}
\label{sec:splines}
In this section, we give a brief introduction to B-splines and the quasi-interpolation operator. We then present the  stability results for the quasi-interpolatory representation of a spline, which plays a key role for our digraph signal processing in Section~\ref{sec:digraph_signal_processing}.

\subsection{Notation}\label{bhag:notation}
We define $x_+=\max(x,0)$ for $x\in\RR$, and for $r>0$, $x_+^r$ by $(x_+)^r$. If $I$ is a real interval, $|I|$ denotes the length of $I$.
If $I=[c-\ell, c+\ell]$, $c\in\RR$, $\ell>0$, and $\alpha>0$, we denote 
\be\label{eq:extendedint}
\alpha\ast I =[c-\alpha\ell, c+\alpha\ell].
\ee
For integer $k\ge 1$,  the space of all polynomials of degree $<k$ is denoted by $\Pi_k$ (so that the dimension of $\Pi_k$ is $k$).

For any non-empty set $K$, (possibly signed) measure $\nu$ on $K$, and $\nu$-measurable function $f :K\to\RR$, we define
\be\label{eq:normdef1}
\|f\|_{\nu,K;p}=\begin{cases}
\disp \left\{\int_K |f(t)|^pd|\nu|(t)\right\}^{1/p}, &\mbox{ if $1\le p <\infty$},\\
\disp |\nu|-\esssup_{t\in K}|f(t)|, &\mbox{if $p=\infty$,}
\end{cases}
\ee
where $|\nu|$ denotes the total variation measure of $\nu$.
When $\nu$ is the Lebesgue measure or $p=\infty$, we will drop the mention of $\nu$. 
Likewise, we will drop the mention of $K$ if $K=\RR$.
If $K$ is a topological space, \textbf{support} of a measure $\nu$, denoted by $\mathsf{supp}(\nu)$ is defined to be the set of all points $x\in K$ such that $|\nu|(B)>0$ for every neighborhood $B$ of $x$. 

For any sequence $\bs a$ (finite or infinite), we define
\be\label{eq:normdef2}
\|\bs a\|_p=\begin{cases}
\disp\left\{\sum_j |a_j|^p\right\}^{1/p}, &\mbox{ if $1\le p <\infty$},\\
\disp \sup_j |a_j|, &\mbox{if $p=\infty$.}
\end{cases}
\ee

In the sequel, $m\ge 1$ is a fixed integer, and $\mu$ is a probability measure supported on $[0,1]$. 
The notation $C\ls D$ denotes $C\le cD$, where $c$ is a generic positive constant depending only on $\mu$, $m$, and the norms.
The notation $C\gs D$ means $D\ls C$ and $C\sim D$ means $C\ls D\ls C$.



\subsection{Splines}\label{bhag:splines}

The material in this section is based on \cite{de1978practical, kunoth2018foundations}, where the notation is different.

For real numbers $t_1\le \cdots\le t_m$ and a function $g$ defined at these points, there exists a unique $P\in \Pi_m$ such that $P(t_k)=g(t_k)$ for $k=1,\ldots, m$, where a derivative interpolation is understood when a knot $t_i$ is repeated. The leading coefficient of this $P$ is denoted by $[t_1,\ldots,t_m]g$.

Let  $\xi_1=0$ and $\xi_{n+1}=1$ and
\be\label{eq:knotsequence}
\bs\xi=(\xi_1=\cdots=\xi_m<\xi_{m+1}<\cdots<\xi_{n+1}=\cdots=\xi_{n+m})
\ee
be a non-decreasing sequence on $[0,1]$. Although we will consider a nested sequence of such sequences later on, we consider $\bs\xi$ to be a fixed sequence for the time being.
The notation means that apart from $\xi_1,\ldots, \xi_{m-1}$, and $\xi_{n+2},\ldots, \xi_{n+m}$, all other knots are distinct.
\begin{definition}\label{def:bspline}
The $B$-spline $B_j$ based on these knots is defined by
\be\label{eq:bsplinedef}
B_j(x)=(\xi_{j+m}-\xi_j)[\xi_j,\ldots,\xi_{j+m}](\xi_j-x)_+^{m-1}, \qquad x\in\RR.
\ee
Let
\be\label{eq:splinespace}
\mathbb{S}=\mathbb{S}_{m,\bs\xi}=\mathsf{span}\{B_j\}_{j=1}^{n+1}.
\ee
A member of $\mathbb{S}_{m,\bs\xi}$ is called a \textbf{spline} of order $m$ with the knot sequence $\bm \xi$. 
\end{definition}

The following proposition summarizes some of the properties of $B$-splines.
\begin{prop}\label{prop:bspline_basic}
{\rm (a)} The restriction $B_{i;j}$ of $B_j$ to any $[\xi_i,\xi_{i+1}]$ is in $\Pi_m$.
For any $i$, the restrictions $B_{i;j}$, $j=i-m+1,\ldots,i$, are a basis for $\Pi_m$.\\
{\rm (b)} $B_j(x)\ge 0$ for all $x\in\RR$. $B_j(x)=0$ if $x\not\in [\xi_j,\xi_{j+m}]$ and $B_j(x)>0$ if $x\in (\xi_j,\xi_{j+m})$.
In particular, for any $j$, only $B_{j-m+1}, \ldots, B_j$ are non-zero on $[\xi_j,\xi_{j+1}]$.\\
{\rm (c)} $\sum_j B_j(x)=1$ for all $x\in [\xi_m,\xi_n]$.\\
The splines $B_1,\ldots, B_n$ are linearly independent on $[\xi_m, \xi_{n+1}]$.\\
\end{prop}


\begin{definition}\label{def:quasi_interp}
An operator $\mathcal{Q}$ of the form
\be\label{eq:quasi_interp_gen}
\mathcal{Q}(f)=\sum_j \lambda_j(f)B_j
\ee
is called a \textbf{quasi-interpolation operator} if \\
{\rm (a)} Each $\lambda_j$ is a linear functional supported on $[\xi_j, \xi_{j+m}]$.\\
{\rm (b)} For every $P\in\Pi_m$, $\mathcal{Q}(P)=P$.
\end{definition}

\begin{prop}\label{prop:spline_reprod}
{\rm (a)} If $\mathcal{Q}$ is a quasi-interpolation operator, where each $\lambda_j$ is supported on only one of the intervals $[\xi_i,\xi_{i+1}]$ for some $i$, has the property that
\be\label{eq:spline_reprod}
\mathcal{Q}(S)=S, \qquad S\in \mathbb{S}.
\ee
In particular, if $S=\sum_j c_jB_j$, then $c_j=\lambda_j(S)$.\\
{\rm (b)} There exists a quasi-interpolatory operator $\mathcal{Q}^*(f)=\sum_j \Lambda_j(f)B_j$, which satisfies the conditions in part (a), so that \eqref{eq:spline_reprod} holds. 
Moreover,
\be\label{eq:specialdeboor}
|\Lambda_j(f)| \ls \|f\|_{[\xi_j,\xi_{j+m}]}, \qquad j=1,\ldots, n.
\ee
\end{prop}


\subsection{Orthogonal polynomials}\label{bhag:orthopoly}
In the remainder of this section, let $\mu$ be a probability measure on $[\xi_m, \xi_{n+1}]\subseteq [0,1]$. For an interval $I\subset [0,1]$, let $\mu_I$ be the restriction of $\mu$ to $I$, normalized again to be a probability measure.
If there are  at least $m+1$ points in $\mathsf{supp}(\mu)\cap I$, then there exists a unique basis $\{p_{k,I}\}_{k=0}^{m-1}$ for $\Pi_m$ such that each $p_{k,I}$ has a positive leading coefficient and
\be\label{eq:orthonormality}
\int_I p_{k,I}p_{j,I}d\mu_I=\frac{1}{\mu(I)}\int_I p_{k,I}p_{j,I}d\mu=\begin{cases}
1, &\mbox{if }j=k,\\
0, &\mbox{otherwise.}
\end{cases}
\ee
We say that $\mu$ is a \textbf{doubling measure} if for any $x\in [0,1]$, $0<y<1/2$,
\be\label{eq:doubling}
\mu\left([x-2y, x+2y]\right)\ls \mu\left([x-y, x+y]\right).
\ee
\begin{lemma}\label{lemma:christbd}
Let $I\subseteq [0,1]$ be an interval, $\mu$ be a doubling measure with at least $m+1$ points in $\mathsf{supp}(\mu)\cap I$. 
Then
\be\label{eq:christbd}
\sum_{k=0}^{m-1}p_{k,I}^2(x)\ls 1, \qquad x\in I.
\ee
\end{lemma}
The proof of this lemma depends upon two well-known facts, summarized in the following proposition.
Part (a) is known as Markov inequality (cf., \cite[Vol. 1, Chapter~VI.\S 6]{natanson}). Part (b) is a simple consequence of the Schwarz inequality and Parseval identity (cf. \cite[Chapter~I, Theorem~4.1]{freudbk}).
Part (c) is a simple consequence of a conformal mapping argument (cf. \cite[Chapter~4, formula~(2.10)]{devlorbk}.
\begin{prop}\label{prop:polyfact}
{\rm (a)} For $P\in\Pi_m$, $-\infty<a<b<\infty$, we have
\be\label{eq:markov}
\|P'\|_{\infty,[a,b]}\le \frac{2m^2}{b-a}\|P\|_{\infty,[a,b]}.
\ee
{\rm (b)} For any $x\in \RR$,
\be\label{eq:freud_extremal}
\max_{\gattha{P\in\Pi_m}{P(x)=1}}\int |P(t)|^2d\mu_I(t)=\sum_{k=0}^{m-1}p_{k,I}^2(x).
\ee
{\rm (c)} If $I\subseteq J\subset \RR$ are compact intervals, then for any $P\in\Pi_m$,
\be\label{eq:walsh}
\|P\|_{\infty,I}\le \|P\|_{\infty, J}\ls (|J|/|I|)^m \|P\|_{\infty,I}.
\ee 
\end{prop}

\vskip 0.5cm
\noindent
\textsc{Proof of Lemma~\ref{lemma:christbd}.}\\
Let  $P\in\Pi_m$ satisfy $\|P\|_{\infty, I}=P(x_0)$ for some $x_0\in I$. 
Since $P^2\in\Pi_{2m}$, the Markov inequality \eqref{eq:markov} implies that
$$
\begin{aligned}
\left||P(t)|^2-|P(x_0)|^2\right|
&\le \frac{8m^2}{|I|}|t-x_0|\|P\|_{\infty,I}^2\\
&= \frac{8m^2}{|I|}|t-x_0||P(x_0)|^2.
\end{aligned}
$$
Consequently, if $I'=((16m^2)^{-1}\ast I)\cap I$, then $|P(t)|^2\ge \|P\|_{\infty,I}^2/2$ for all $t\in I'$.
Hence,
\be\label{eq:pf1eqn1}
\int_I |P(t)|^2d\mu_I(t)\ge \int_{I'} |P(t)|^2d\mu_I(t)\ge \frac{\|P\|_{\infty,I}^2}{2}\mu_I(I').
\ee
Since $\mu$ is a doubling measure, we observe that
$$
\mu_I(I')=\frac{\mu(I')}{\mu(I)}\gs 1.
$$
So, we have proved that for any $P\in\Pi_m$,
$$
\int_I |P(t)|^2d\mu_I(t)\gs \|P\|_{\infty,I}^2.
$$
The extremal principle \eqref{eq:freud_extremal} now leads to  the estimate \eqref{eq:christbd}.
\qed

\subsection{Stability theorem}\label{bhag:stability}
The purpose of this section is to prove the following generalization of the well-known de Boor stability theorem.
\begin{theorem}\label{theo:stabilitytheo}
Let $S=\sum_{j\in\ZZ}c_jB_j\in \mathbb{S}$, and $\mu$ be a doubling measure. then for $1\le p<\infty$,
\be\label{eq:stabilitytheo}
\sum_{j\in\ZZ}\mu(I_j)|c_j|^p\sim \int_{\xi_m}^{\xi_{n+1}} |f(t)|^pd\mu(t),
\ee
with an obvious modification in the case when $p=\infty$.
\end{theorem}
Following \cite{kunoth2018foundations}, we first obtain a quasi-interpolatory representation of a spline.
For $m\le j\le n$, we denote $I_j=[\xi_j,\xi_{j+1}]$, and 
Let $j\le j^*<j+m$ be chosen so that 
$$
|I_{j^*}|=\max_{j\le \ell<j+m}|I_\ell|.
$$
We note that
\be\label{eq:pf2eqn1}
I_j^*\subseteq [\xi_j, \xi_{j+m}]\subset (2m)\ast I_{j^*}.
\ee
With the notation as in Proposition~\ref{prop:spline_reprod}(b),  we define
\be\label{eq:qifunctionaldef}
\begin{aligned}
\lambda_j(f)=\int_{I_{j^*}} \left\{\sum_{k=0}^{m-1} \Lambda_j(p_{k,I_{j^*}})p_{k,I_{j^*}}(x)\right\}f(x)d\mu_{I_{j^*}}(x), \\
j=1,\cdots,n+1.
\end{aligned}
\ee
and
\be\label{eq:quasi_interp_op}
\mathcal{Q}(f)(x)=\sum_j \lambda_{j}(f)B_j(x), \qquad x\in [x_m,x_{n+1}].
\ee
\begin{lemma}\label{lemma:specialquasi}
{\rm (a)} $\mathcal{Q}$ defined in \eqref{eq:quasi_interp_op} is a quasi-interpolation operator in the sense of Definition~\ref{def:quasi_interp}. \\
{\rm (b)} We have $Q(S)=S$ for all $S\in\mathbb{S}$.\\
{\rm (c)} For  $1\le p<\infty$, we have
\be\label{eq:lambdajest}
\sum_{j=m}^{n+1} \mu(I_j)|\lambda_{j}(f)|^p \ls \int_{\xi_m}^{\xi_{n+1}}|f(x)|^pd\mu(x).
\ee
For $p=\infty$, we have
\be\label{eq:lambdajest_inf}
\max_{m\le j\le n+1}\lambda_j(f)|\ls \|f\|_{\infty;[\xi_m,\xi_{n+1}]}.
\ee
\end{lemma}

\noindent\textsc{Proof of Lemma~\ref{lemma:specialquasi}}\\
Obviously,
$$
\lambda_{j}(p_{\ell,I_{j^*}})=\Lambda_j(p_{\ell,I_{j^*}}), \qquad \ell=0,\cdots, m-1.
$$
It follows from Proposition~\ref{prop:spline_reprod}(b) that
$$
\mathcal{Q}(P)=P, \qquad P\in\Pi_m.
$$
Thus, $\mathcal{Q}$ is a quasi-interpolation operator. This proves part (a).

Since $\Lambda_{j}$ is supported on only one interval  $I_{j^*}$, Proposition~\ref{prop:spline_reprod} implies part (b).

In this part of the proof, let $R_x(y)=\sum_{k=0}^{m-1} p_{k,I_{j^*}}(y)p_{k,I_{j^*}}(x)$. 
The Schwarz inequality yields that
$$
|R_x(y)|\le \left\{\sum_{k=0}^{m-1}p_{k,I_{j^*}}(y)^2\right\}^{1/2}\left\{\sum_{k=0}^{m-1}p_{k,I_{j^*}}(x)^2\right\}^{1/2},
$$
so that, in view of Lemma~\ref{lemma:christbd},
$$
\max_{x\in I_{j^*}}\|R_x\|_{\infty, I_{j^*}}\ls 1.
$$
In view of Proposition~\ref{prop:polyfact}(c) and \eqref{eq:pf2eqn1}, we deduce that
$$
\max_{x\in I_{j^*}}\|R_x\|_{\infty, [\xi_j,\xi_{j+m}]}\ls 1.
$$
Consequently, \eqref{eq:specialdeboor} implies that for all $x\in I_{j^*}$,
$$
\left|\sum_{k=0}^{m-1} \Lambda_j(p_{k,I_{j^*}})p_{k,I_{j^*}}(x)\right|=\left|\Lambda_j\left(R_x\right)\right|\ls 1.
$$
Therefore, a straightforward application of H\"older inequality in the definition \eqref{eq:qifunctionaldef}, taking into account the fact that $\mu$ is a doubling measure, implies that for $1\le p<\infty$,
\be\label{eq:locstability}
\begin{aligned}
|\lambda_j(f)|^p&\ls \int_{I_j^*}|f(x)|^pd\mu_{I_j^*}(x)=\frac{1}{\mu(I_{j^*})}\int_{I_j^*}|f(x)|^pd\mu(x)\\
&\ls \frac{1}{\mu(I_j)}\int_{I_j^*}|f(x)|^pd\mu(x).
\end{aligned}
\ee
The estimate \eqref{eq:lambdajest} is obtained by adding these inequalities, $j=m,\cdots, n+1$.
The estimate \eqref{eq:lambdajest_inf} is simpler. \qed

\smallskip

\noindent\textsc{Proof of Theorem~\ref{theo:stabilitytheo}}\\
Since $\{B_j\}$ is a basis for $\mathbb{S}$, it follows from Lemma~\ref{lemma:specialquasi}(b) that for $S=\sum_j c_jB_j$, $c_j=\lambda_j(S)$. 
Hence, Lemma~\ref{lemma:specialquasi}(c) leads immediately to the bound in \eqref{eq:stabilitytheo} for estimating the discrete norm of $\mathbf{c}=\{c_j\}$ in terms of the norm of $S$. 
In the reverse direction, we note that since $B_j\ge 0$ and $\sum_jB_j\equiv 1$,
\be\label{eq:infstability}
\|S\|_\infty \le \left\|\sum_jc_jB_j\right\|\le \|\mathbf{c}\|_\infty.
\ee
Further, since each $B_j$ is supported on $[\xi_j, \xi_{j+m}]$, $0\le B_j\le 1$, and $\mu$ is a doubling measure (cf. \eqref{eq:pf2eqn1}),
$$
\int_{\xi_m}^{\xi_{n+1}}|s|d\mu\le \sum_{j}|c_j|\int_{\xi_j}^{\xi_{j+m}}B_jd\mu\lesssim \sum_j|c_j|\mu(I_j).
$$
Together with \eqref{eq:infstability}, this proves the other inequality in \eqref{eq:stabilitytheo} for $p=1,\infty$. The general case follows by applying the Riesz-Thorin interpolation theorem. \qed

\section{Digraph Signal Processing via Spline Quasi-Interpolants}
\label{sec:digraph_signal_processing}
The hierarchical partitions produced by the \myproj~algorithm induce a hierarchy of block partitions. In this section, we shall apply this hierarchy to construct the splines and their associated quasi-interpolation operators for digraph signal processing. 

\subsection{Digraph Signal Processing}
Specifically, given a digraph signal $f: V\rightarrow \mathbb R$ defined on a vertex set $V$ of a digraph $G=(V,E,A)$, we first apply our \myproj~algorithm to form a hierarchical tree (filtration) and then it induces the hierarchy $\mathcal{R}_{0}\succeq\mathcal{R}_1\succeq\cdots\succeq\mathcal{R}_L\succeq \mathcal{R}_{L+1}$ of blocks, where each block in $\mathcal R_l$ is given by a pair of directional intervals. At each level $l$, we convert $\mathcal R_l$ into a pair  $(\bm \xi_{l,\mathrm{in}},\bm \xi_{l,\mathrm{out}})$ of knot sequences, which are then used to construct the pair $(\{B^{(l)}_{j,\mathrm{in}}\}_{j=0}^{k_l},\{B^{(l)}_{j,\mathrm{out}}\}_{j=0}^{k_l})$ of spline bases and the corresponding pair $(\mathcal Q^{(l)}_{\mathrm{in}},\mathcal Q^{(l)}_{\mathrm{out}})$  of quasi-interpolants at level $l$, as described in Section~\ref{sec:splines}. Due to the hierarchy property, the knot sequences are nested for each direction. The differences $\mathcal Q^{(l+1)}_{\iota}f-\mathcal Q^{(l)}_{\iota}f$, $\iota\in\{\mathrm{in,out}\}$, between successive approximations provide detailed information about the graph signal $f$ and can be used for graph signal denoising.

\noindent\textbf{Knot Sequences from Multi-Level Clusters}. Fix $m$ to be the dimension of the polynomial space $\Pi_m$. 
Let $\mathcal{R}_l=\{R_j^{(l)}={I}^{(l)}_{j,\mathrm{in}}\times {I}^{(l)}_{j,\mathrm{out}}: j=1,\ldots,k_l\}$ be the rectangular blocks at level $l$. Note that $[0,1]=\cup_{j=1}^{k_l}I^{(l)}_{j,\iota}$ for any $\iota\in\{\mathrm{in, out}\}$ and any $l=0,\ldots,L+1$.
We collect and order the distinct endpoints of $\{I_{j,\iota}: j=1,\ldots, k_l\}$ in Section~\ref{sec:interval-partition} as
\[
\bm\xi_{l,\iota}=(\xi^{(l)}_{0,\iota}=\cdots=\xi^{(l)}_{0,\iota}<\xi^{(l)}_{1,\iota}<\ldots<\xi^{(l)}_{k_l,\iota}=\cdots=\xi^{(l)}_{k_l,\iota})
\]
for $\iota\in\{\mathrm{in,out}\}$,
where $\xi^{(l)}_{0,\iota}=0$ and $\xi^{(l)}_{k_l,\iota}=1$ are repeated $m+1$ times for the construction of B-splines of order $m$. Since the intervals at level $l$ are obtained by subdividing those at level $l-1$, the knot sequences satisfy
\begin{equation}
\bm\xi_{0,\iota}\subseteq\bm\xi_{1,\iota}\subseteq\cdots\subseteq \bm\xi_{L,\iota}\subseteq \bm\xi_{L+1,\iota}.
\end{equation}
Each knot sequence  $\bm\xi_{l,\iota}$ then induces a spline basis $\{B_{j,\iota}^{(l)}\}_{j=0}^{k_l}$ for the space $\mathbb S_{m,\bm \xi_{l,\iota}}$. 

For the finest level $L+1$, we associate vertex $v_i$ with the coordinate $x_{i,\iota}=\xi^{(L+1)}_{i,\iota}$, $i=1,\ldots,N$.
The graph signal $f:V\rightarrow \mathbb R$ can be represented as a spline function $f_{\iota}(t)=\sum_{j=0}^{k_l}c_{j,\iota}^{(L+1)}B^{(L+1)}_{j,\iota}(t)\in \mathbb S_{m,\bm \xi_{L+1,\iota}}$ for $t\in[0,1]$ via the spline interpolation algorithm \cite{de1978practical}. That is, $f_\iota(x_{i,\iota}) = f(v_i)$ for $\iota\in\{\mathrm{in,out}\}$. 


\noindent\textbf{Spline Quasi-Interpolation for Graph Signals.}
Given $f_\iota$ and the spline basis $\{B_{j,\iota}^{(l)}\}_{j=0}^{k_l}$, by Definition~\ref{def:quasi_interp}, the corresponding spline quasi-interpolant for $f_\iota$ at level $l$ is given by
\begin{equation}
\mathcal{Q}^{(l)}_\iota f_\iota(t)=\sum_{j=0}^{k_l}\lambda_{j,\iota}^{(l)}(f_\iota)B^{(l)}_{j,\iota}(t),\quad l=1,\ldots,L+1,
\end{equation}
where $\lambda^{(l)}_{j,\iota}$ is a local linear functional determined by the signal samples within the support of $B^{(l)}_{j,\iota}$. The compact support of the B-spline basis allows the coefficients to be computed locally without solving a global interpolation system (see \cite{kunoth2018foundations}). Evaluating the quasi-interpolant $\mathcal Q_\iota^{(l)} f_\iota$ at the (finest) vertex coordinates $x_{i,\iota}$ gives a quasi-interpolant graph signal $\bm{L}_{l,\iota}$ as 
\begin{equation}
\bm{L}_{l,\iota}=\left[\mathcal{Q}^{(l)}_\iota f(x_{1,\iota}),\ldots,\mathcal{Q}^{(l)}_\iota f(x_{N,\iota})\right]^\top.
\end{equation}
Note that a finer knot sequence captures more local variations, whereas a coarser sequence produces a smoother structural approximation. 

\begin{table*}[htpb!]
  \centering
  \caption{Denoising performance of graph signals under varying noise ratios, comparing raw corrupted signals with spline-smoothed reconstructions using multi-resolution hierarchical structures.}\label{Table:spline}
  \scalebox{0.90}[0.90]{
    \begin{tabular}{lllrrrrr}
      \toprule
       \textbf{Noise Level}  & \textbf{Signal Pairs}  & \textbf{Metrics}  & \textbf{Cora} & \textbf{Cornell} & \textbf{Texas} & \textbf{Wisconsin} & \textbf{Squirrel} \\ \midrule
        \multirow{6}{*}{5\%} & \multirow{2}{*}{$\bm y_T$ and $\bm y_O$} 
           & RMSE  & 0.0053  & 0.0078  & 0.0078  & 0.0068  & 0.0046  \\
           &   & SNR     & 26.1595  & 26.4977  & 26.4977 & 26.3257 & 26.0635  \\
         & \multirow{2}{*}{$\bm y_T$ and $\bm L_{2}+ \tilde{\bm d}_2$ } 
           & RMSE  & 0.0041   & 0.0056   & 0.0063  & 0.0044 & 0.0023  \\
           &   & SNR    & 28.2786  & 29.3331  & 28.3221 & 30.0617 & 31.9974  \\
         & \multirow{2}{*}{$\bm y_T$ and $ \bm L_{1}+\tilde{\bm d}_1 + \tilde{\bm d}_2$ } 
           & RMSE & 0.0035   & 0.0056   & 0.0063  & 0.0044 & 0.0023  \\
           &    & SNR     & 29.6455  & 29.3331  & 28.3221 & 30.0617 & 31.9974  \\
           \toprule
       \textbf{Noise Level}  & \textbf{Signal Pairs}  & \textbf{Metrics}  & \textbf{Cora} & \textbf{Cornell} & \textbf{Texas} & \textbf{Wisconsin} & \textbf{Squirrel} \\ \midrule
        \multirow{6}{*}{10\%} & \multirow{2}{*}{$\bm y_T$ and $\bm y_O$} 
           & RMSE & 0.0080   & 0.0156   & 0.0156  & 0.0136  & 0.0073  \\
           &   & SNR     & 20.1389  & 20.4771  & 20.4771 & 20.3051 & 20.0429  \\
         & \multirow{2}{*}{$\bm y_T$ and $\bm L_{2}+ \tilde{\bm d}_2$ }  
           & RMSE  & 0.0082   & 0.0113   & 0.0127  & 0.0089 & 0.0037  \\
           &   & SNR     & 22.2583  & 23.3126  & 22.3091 & 24.0513 & 26.0014  \\
         & \multirow{2}{*}{$\bm y_T$ and $ \bm L_{1}+\tilde{\bm d}_1 + \tilde{\bm d}_2$ } 
           & RMSE  & 0.0070   & 0.0113   & 0.0127  & 0.0089 & 0.0037  \\
           &    & SNR     & 23.6250  & 23.3126  & 22.3091 & 24.0513 & 26.0014  \\ 
         \toprule
       \textbf{Noise Level}  & \textbf{Signal Pairs}  & \textbf{Metrics}  & \textbf{Cora} & \textbf{Cornell} & \textbf{Texas} & \textbf{Wisconsin} & \textbf{Squirrel} \\ \midrule
         \multirow{6}{*}{15\%} & \multirow{2}{*}{$\bm y_T$ and $\bm y_O$} 
           & RMSE  & 0.0158  & 0.0235  & 0.0235  & 0.0204  & 0.0139  \\
           &   & SNR     & 16.6171  & 16.9552  & 16.9552 & 16.7833 & 16.5211  \\
          & \multirow{2}{*}{$\bm y_T$ and $\bm L_{2}+ \tilde{\bm d}_2$ } 
           & RMSE  & 0.0124  & 0.0169  & 0.0190  & 0.0133 & 0.0070 \\
           &   & SNR     & 18.7367  & 19.7906  & 18.7943 & 20.5393 & 22.4556  \\
         & \multirow{2}{*}{$\bm y_T$ and $ \bm L_{1}+\tilde{\bm d}_1 + \tilde{\bm d}_2$ } 
           & RMSE  & 0.0106   & 0.0169   & 0.0190  & 0.0133 & 0.0070  \\
           &    & SNR     & 20.1034  & 19.7906  & 18.7943 & 20.5393 & 22.4556  \\ 
           \toprule
       \textbf{Noise Level}  & \textbf{Signal Pairs}  & \textbf{Metrics}  & \textbf{Cora} & \textbf{Cornell} & \textbf{Texas} & \textbf{Wisconsin} & \textbf{Squirrel} \\ \midrule
        \multirow{6}{*}{20\%} & \multirow{2}{*}{$\bm y_T$ and $\bm y_O$} 
           & RMSE  & 0.0211  & 0.0313  & 0.0313  & 0.0273  & 0.0185 \\
           &   & SNR     & 14.1183  & 14.4565  & 14.4565 & 14.2845 & 14.0223  \\
        & \multirow{2}{*}{$\bm y_T$ and $\bm L_{2}+ \tilde{\bm d}_2$ } 
           & RMSE  & 0.0165   & 0.0226  & 0.0253  & 0.0177 & 0.0093 \\
           &   & SNR    & 16.2381  & 17.2915  & 16.3010 & 18.0501 & 19.9570  \\
     & \multirow{2}{*}{$\bm y_T$ and $ \bm L_{1}+\tilde{\bm d}_1 + \tilde{\bm d}_2$ }
           & RMSE  & 0.0141   &  0.0226  & 0.0253  & 0.0177 & 0.0093  \\
           &    & SNR    & 17.6025  & 17.2915  & 16.3010 & 18.0501 & 19.9570  \\ \bottomrule      
\end{tabular}}
\end{table*}

\noindent\textbf{Multi-level Digraph Signal Denoising}.
Let $\bm{y}_T=[f(v_1),\ldots,f(v_N)]^\top$ denote the underlying clean digraph signal (ground truth). Its noisy observation is modeled as
\begin{equation}
\bm{y}_O=\bm{y}_T+\sigma\bm{\varepsilon},
\end{equation}
where $\sigma=\max(|y_T|)r$, $r$ is the (relative) noise level, and $\bm{\varepsilon}$ is a standard Gaussian noise vector satisfying $\bm{\varepsilon}\sim \mathcal{N}(0,1)$. Only $\bm{y}_O$ is available during reconstruction, while $\bm{y}_T$ is used solely to evaluate the recovery quality. Applying the spline quasi-interpolation operator to $\bm{y}_O$ at each resolution gives the approximation vectors $\bm{L}_{1,\iota},\ldots,\bm{L}_{L+1,\iota}$ in $\mathbb R^{N}$ for $\iota\in\{\mathrm{in,out}\}$. The detail between two consecutive resolutions is defined as
$\bm{d}_{j,\iota}=\bm{L}_{j+1,\iota}-\bm{L}_{j,\iota}$, $j=1,\ldots,L$.
Each $\bm{d}_{j,\iota}$ contains the details at resolution $j$. Note that for any
$j_0\in\{2,\ldots,J\}$, we have the decomposition
\begin{equation}
\bm{L}_{L+1,\iota}=\bm{L}_{j_0,\iota}+\sum_{j=j_0}^{L}\bm{d}_{j,\iota}.
\end{equation}
For the noisy observation, these detail components $\bm d_{j,\iota}$ are typically sparse and contain both useful local features and noise. To suppress noise, we apply an adaptive thresholding operator $\mathcal{T}_j$ on $\bm d_{j,\iota}$ to obtain
$\widetilde{\bm{d}}_{j,\iota}=\mathcal{T}_j(\bm{d}_{j,\iota})$, $j=1,\ldots,L$.
The threshold is adjusted according to the estimated noise level and local residual activity, thereby suppressing weak noisy fluctuations while preserving significant signal variations. From $\bm{L}_{j_0,N}$ as the coarse structural component, the denoised signal is reconstructed as
\begin{equation}
\widetilde{\bm{y}}_{j_0:L+1,\iota}=\bm{L}_{j_0,\iota}+\sum_{j=j_0}^{L}\widetilde{\bm{d}}_{j,\iota},\quad j_0=1,\ldots,L.
\end{equation}
The final denoised signal is given by 
\[
\widetilde{\bm{y}}_{j_0:L+1}=\frac12(\widetilde{\bm{y}}_{j_0:L+1,\mathrm{in}}+\widetilde{\bm{y}}_{j_0:L+1,\mathrm{out}}).
\]
That is, the average of the denoised results with respect to the in-degree and out-degree spline quasi-interpolants. The objective is to obtain $\widetilde{\bm{y}}_{j_0:L+1}\approx\bm{y}_T$ from the noisy observation $\bm{y}_O$. In this process, the coarse quasi-interpolant captures the dominant smooth structure, while the thresholded details restore informative variations at finer resolutions.

\subsection{Denoising Performance Analysis}
We evaluate the hierarchy-induced knot sequences on graph signal denoising.
Let $\bm y_T$ denote the ground-truth signal and $\bm y_N$ its noisy observation,
with noise ratios $r$ ranging from $5\%$ to $20\%$. Following the
multi-level reconstruction defined above, we report the results from the one-level denoising
$\bm L_{2,\iota}+\widetilde{\bm d}_{2,\iota}$ and the two-level denoising $\bm L_{1,\iota}+\widetilde{\bm d}_{1,\iota}+\widetilde{\bm d}_{2,\iota}$. Performance is measured by RMSE (root mean square error, the smaller the better)
and SNR (signal to noise ratio, the larger the better) between $\bm Y_T$ and $\widetilde{\bm{y}}_{j_0:L+1}$ for $L=2$ and $j_0=1,2$. 

As shown in Table~\ref{Table:spline}, increasing the noise ratio consistently degrades the raw signal, whereas spline-based reconstruction generally reduces RMSE and improves SNR across all datasets. At a $20\%$ noise ratio, the best reconstruction reduces the RMSE from $0.0211$ to $0.0141$ on Cora and from $0.0185$ to $0.0093$ on Squirrel, with corresponding SNR improvements from $14.1183$ to $17.6025$ dB and from $14.0223$ to $19.9570$ dB. The three-level reconstruction further improves the results on Cora, indicating
that the additional detail component contains useful fine-scale information. On the other datasets, the two reconstruction depths yield similar or identical results, suggesting that the coarser approximation already captures most of the
recoverable signal structure. Overall, these results confirm that the cluster-derived knot sequences provide effective and stable supports for structure-adaptive graph signal denoising.

\section{Conclusions}
\label{sec:conclusions}
This work connected hierarchical digraph clustering with spline-based graph signal processing. We proposed \myproj, a semi-supervised spectral
hierarchical clustering framework based on a generalized Hermitian matrix. By recursively clustering and coarsening the digraph, \myproj~preserves directional interactions and produces consistent nested partitions under both supervised and unsupervised settings. The resulting hierarchy is mapped to in-degree and out-degree intervals, whose endpoints form nested nonuniform knot sequences. These knots define multilevel spline quasi-interpolants, allowing noisy graph signals to be reconstructed from a coarse approximation and adaptively thresholded inter-level details. Experiments on synthetic and real-world digraphs demonstrate the effectiveness
of \myproj~under different structural and supervision settings. The denoising results further show that the hierarchy-induced knots provide structure-adaptive supports for signal recovery, yielding lower RMSE and higher SNR. Future work will focus on scalable implementations and extensions to more complex directed graph structures.
\bibliographystyle{IEEEtran}
\bibliography{ref}

@article{sandryhaila2013discrete,
  title={Discrete signal processing on graphs},
  author={Sandryhaila, Aliaksei and Moura, Jos{\'e} MF},
  journal={IEEE Transactions on Signal Processing},
  volume={61},
  number={7},
  pages={1644--1656},
  year={2013},
  publisher={IEEE}
}

@article{gama2019convolutional,
  title={Convolutional neural network architectures for signals supported on graphs},
  author={Gama, Fernando and Marques, Antonio G and Leus, Geert and Ribeiro, Alejandro},
  journal={IEEE Transactions on Signal Processing},
  volume={67},
  number={4},
  pages={1034--1049},
  year={2019},
  publisher={IEEE}
}

@article{leus2023graph,
  title={Graph signal processing: {H}istory, development, impact, and outlook},
  author={Leus, Geert and Marques, Antonio G and Moura, Jos{\'e} MF and Ortega, Antonio and Shuman, David I},
  journal={IEEE Signal Processing Magazine},
  volume={40},
  number={4},
  pages={49--60},
  year={2023},
  publisher={IEEE}
}

@article{ortega2018graph,
  title={Graph signal processing: {O}verview, challenges, and applications},
  author={Ortega, Antonio and Frossard, Pascal and Kova{\v{c}}evi{\'c}, Jelena and Moura, Jos{\'e} MF and Vandergheynst, Pierre},
  journal={Proceedings of the IEEE},
  volume={106},
  number={5},
  pages={808--828},
  year={2018},
  publisher={IEEE}
}

@article{von2007tutorial,
  title={A tutorial on spectral clustering},
  author={Von Luxburg, Ulrike},
  journal={Statistics and Computing},
  volume={17},
  number={4},
  pages={395--416},
  year={2007},
  publisher={Springer}
}

@article{shuman2020localized,
  title={Localized spectral graph filter frames: {A} unifying framework, survey of design considerations, and numerical comparison},
  author={Shuman, David I},
  journal={IEEE Signal Processing Magazine},
  volume={37},
  number={6},
  pages={43--63},
  year={2020},
  publisher={IEEE}
}

@article{jung2019localized,
  title={Localized linear regression in networked data},
  author={Jung, Alexander and Tran, Nguyen},
  journal={IEEE Signal Processing Letters},
  volume={26},
  number={7},
  pages={1090--1094},
  year={2019},
  publisher={IEEE}
}

@article{jung2019semi,
  title={Semi-supervised learning in network-structured data via total variation minimization},
  author={Jung, Alexander and Hero III, Alfred O and Mara, Alexandru Cristian and Jahromi, Saeed and Heimowitz, Ayelet and Eldar, Yonina C},
  journal={IEEE Transactions on Signal Processing},
  volume={67},
  number={24},
  pages={6256--6269},
  year={2019},
  publisher={IEEE}
}

@article{tanaka2020sampling,
  title={Sampling signals on graphs: {F}rom theory to applications},
  author={Tanaka, Yuichi and Eldar, Yonina C and Ortega, Antonio and Cheung, Gene},
  journal={IEEE Signal Processing Magazine},
  volume={37},
  number={6},
  pages={14--30},
  year={2020},
  publisher={IEEE}
}

@article{tanaka2020generalized,
  title={Generalized sampling on graphs with subspace and smoothness priors},
  author={Tanaka, Yuichi and Eldar, Yonina C},
  journal={IEEE Transactions on Signal Processing},
  volume={68},
  pages={2272--2286},
  year={2020},
  publisher={IEEE}
}

@article{rosvall2008maps,
  title={Maps of random walks on complex networks reveal community structure},
  author={Rosvall, Martin and Bergstrom, Carl T},
  journal={Proceedings of the National Academy of Sciences},
  volume={105},
  number={4},
  pages={1118--1123},
  year={2008},
  publisher={National Academy of Sciences}
}

@inproceedings{cucuringu2020hermitian,
  title={Hermitian matrices for clustering directed graphs: {I}nsights and applications},
  author={Cucuringu, Mihai and Li, Huan and Sun, He and Zanetti, Luca},
  booktitle={International Conference on Artificial Intelligence and Statistics},
  pages={983--992},
  year={2020},
  organization={PMLR}
}

@inproceedings{satuluri2011symmetrizations,
  title={Symmetrizations for clustering directed graphs},
  author={Satuluri, Venu and Parthasarathy, Srinivasan},
  booktitle={Proceedings of the 14th International Conference on Extending Database Technology},
  pages={343--354},
  year={2011}
}

@article{zhang2022directed,
  title={Directed community detection with network embedding},
  author={Zhang, Jingnan and He, Xin and Wang, Junhui},
  journal={Journal of the American Statistical Association},
  volume={117},
  number={540},
  pages={1809--1819},
  year={2022},
  publisher={Taylor \& Francis}
}

@article{chung2005laplacians,
  title={Laplacians and the Cheeger inequality for directed graphs},
  author={Chung, Fan},
  journal={Annals of Combinatorics},
  volume={9},
  number={1},
  pages={1--19},
  year={2005},
  publisher={Springer}
}

@inproceedings{zhou2005learning,
  title={Learning from labeled and unlabeled data on a directed graph},
  author={Zhou, Dengyong and Huang, Jiayuan and Sch{\"o}lkopf, Bernhard},
  booktitle={Proceedings of the International Conference on Machine Learning (ICML)},
  pages={1036--1043},
  year={2005}
}

@inproceedings{he2022digrac,
  title={Digrac: {D}igraph clustering based on flow imbalance},
  author={He, Yixuan and Reinert, Gesine and Cucuringu, Mihai},
  booktitle={Proceedings of Learning on Graphs Conference},
  year={2022},
  organization={PMLR}
}

@inproceedings{laenen2020higher,
  title={Higher-order spectral clustering of directed graphs},
  author={Laenen, Steinar and Sun, He},
  booktitle={Advances in Neural Information Processing Systems},
  pages={941--951},
  year={2020}
}

@book{deboor2001practical,
  title={A Practical Guide to Splines},
  author={De Boor, Carl},
  year={2001},
  series={Applied Mathematical Sciences},
  volume={27},
  publisher={Springer-Verlag},
  address={New York}
}

@article{de1973spline,
  title={Spline approximation by quasiinterpolants},
  author={De Boor, Carl and Fix, George J},
  journal={Journal of Approximation Theory},
  volume={8},
  number={1},
  pages={19--45},
  year={1973},
  publisher={Academic Press}
}

@article{chui2015representation,
  title={Representation of functions on big data: graphs and trees},
  author={Chui, Charles K and Filbir, Frank and Mhaskar, Hrushikesh N},
  journal={Applied and Computational Harmonic Analysis},
  volume={38},
  number={3},
  pages={489--509},
  year={2015},
  publisher={Elsevier}
}

@article{chui2016multirate,
  title={Multirate systems with shortest spline-wavelet filters},
  author={Chui, Charles K and De Villiers, Johan and Zhuang, Xiaosheng},
  journal={Applied and Computational Harmonic Analysis},
  volume={41},
  number={1},
  pages={266--296},
  year={2016},
  publisher={Elsevier}
}

@article{chui2018representation,
  title={Representation of functions on big data associated with directed graphs},
  author={Chui, Charles K and Mhaskar, HN and Zhuang, Xiaosheng},
  journal={Applied and Computational Harmonic Analysis},
  volume={44},
  number={1},
  pages={165--188},
  year={2018},
  publisher={Elsevier}
}

@article{speleers2017hierarchical,
  title={Hierarchical spline spaces: {Q}uasi-interpolants and local approximation estimates},
  author={Speleers, Hendrik},
  journal={Advances in Computational Mathematics},
  volume={43},
  number={2},
  pages={235--255},
  year={2017},
  publisher={Springer}
}

@article{bracco2016bivariate,
  title={Bivariate hierarchical {H}ermite spline quasi-interpolation},
  author={Bracco, Cesare and Giannelli, Carlotta and Mazzia, Francesca and Sestini, Alessandra},
  journal={BIT Numerical Mathematics},
  volume={56},
  number={4},
  pages={1165--1188},
  year={2016},
  publisher={Springer}
}

@article{huang2024flow2gnn,
  title={{Flow2GNN}: {F}lexible two-way flow message passing for enhancing gnns beyond homophily},
  author={Huang, Changqin and Wang, Yi and Jiang, Yunliang and Li, Ming and Huang, Xiaodi and Wang, Shijin and Pan, Shirui and Zhou, Chuan},
  journal={IEEE Transactions on Cybernetics},
  volume={54},
  number={11},
  pages={6607--6618},
  year={2024},
  publisher={IEEE}
}

@article{gates2017impact,
  title={The impact of random models on clustering similarity},
  author={Gates, Alexander J and Ahn, Yong-Yeol},
  journal={Journal of Machine Learning Research},
  volume={18},
  number={87},
  pages={1--28},
  year={2017}
}

@article{malliaros2013clustering,
  title={Clustering and community detection in directed networks: {A} survey},
  author={Malliaros, Fragkiskos D and Vazirgiannis, Michalis},
  journal={Physics Reports},
  volume={533},
  number={4},
  pages={95--142},
  year={2013},
  publisher={Elsevier}
}

@article{rohe2016co,
  title={Co-clustering directed graphs to discover asymmetries and directional communities},
  author={Rohe, Karl and Qin, Tai and Yu, Bin},
  journal={Proceedings of the National Academy of Sciences},
  volume={113},
  number={45},
  pages={12679--12684},
  year={2016},
  publisher={National Academy of Sciences}
}

@article{shi2000normalized,
  title={Normalized cuts and image segmentation},
  author={Shi, Jianbo and Malik, Jitendra},
  journal={IEEE Transactions on Pattern Analysis and Machine Intelligence},
  volume={22},
  number={8},
  pages={888--905},
  year={2000},
  publisher={Ieee}
}

@article{flake2004graph,
  title={Graph clustering and minimum cut trees},
  author={Flake, Gary William and Tarjan, Robert E and Tsioutsiouliklis, Kostas},
  journal={Internet Mathematics},
  volume={1},
  number={4},
  pages={385--408},
  year={2004},
  publisher={Taylor \& Francis}
}

@article{ng2001spectral,
  title={On spectral clustering: {A}nalysis and an algorithm},
  author={Ng, Andrew and Jordan, Michael and Weiss, Yair},
  journal={Advances in Neural Information Processing Systems},
  volume={14},
  year={2001}
}

@article{wei2025vertex,
  title={Vertex-frequency analysis on directed graphs},
  author={Wei, Deyun and Yuan, Shuangxiao},
  journal={IEEE Transactions on Signal Processing},
  volume={73},
  pages={2255--2270},
  year={2025},
  publisher={IEEE}
}

@article{pesenson2009variational,
  title={Variational splines and Paley--Wiener spaces on combinatorial graphs},
  author={Pesenson, Isaac},
  journal={Constructive Approximation},
  volume={29},
  number={1},
  pages={1--21},
  year={2009},
  publisher={Springer}
}

@inproceedings{hayashi2022skew,
  title={Skew-symmetric adjacency matrices for clustering directed graphs},
  author={Hayashi, Koby and Aksoy, Sinan G and Park, Haesun},
  booktitle={Proceedings of the IEEE International Conference on Big Data},
  pages={555--564},
  year={2022},
  organization={IEEE}
}

@article{dahmen1980multidimensional,
  title={Multidimensional spline approximation},
  author={Dahmen, Wolfgang and De Vore, R and Scherer, Karl},
  journal={SIAM Journal on Numerical Analysis},
  volume={17},
  number={3},
  pages={380--402},
  year={1980},
  publisher={SIAM}
}

@incollection{cox2006practical,
  title={Practical spline approximation},
  author={Cox, Maurice G},
  booktitle={Topics in Numerical Analysis: Proceedings of the SERC Summer School},
  pages={79--112},
  year={2006},
  publisher={Springer}
}

@article{hasan2024b,
  title={B-spline curve theory: {A}n overview and applications in real life},
  author={Hasan, Md Shahid and Alam, Md Nur and Fayz-Al-Asad, Md and Muhammad, Noor and Tun{\c{c}}, Cemil},
  journal={Nonlinear Engineering},
  volume={13},
  number={1},
  pages={20240054},
  year={2024},
  publisher={De Gruyter}
}

@article{sen2008collective,
  title={Collective Classification in Network Data},
  author={Sen, Prithviraj and Namata, Galileo and Bilgic, Mustafa and Getoor, Lise and Galligher, Brian and Eliassi-Rad, Tina},
  journal={AI Magazine},
  volume={29},
  number={3},
  pages={93--93},
  year={2008}
}

@inproceedings{pei2020geom,
  title={Geom-gcn: Geometric Graph Convolutional Networks},
  author={Pei, Hongbin and Wei, Bingzhe and Chang, Kevin Chen-Chuan and Lei, Yu and Yang, Bo},
  booktitle={International Conference on Learning Representations},
  year={2020}
}

@article{rozemberczki2021multi,
  title={Multi-scale attributed node embedding},
  author={Rozemberczki, Benedek and Allen, Carl and Sarkar, Rik},
  journal={Journal of Complex Networks},
  volume={9},
  number={2},
  pages={cnab014},
  year={2021},
  publisher={Oxford University Press}
}

@inproceedings{adamic2005political,
  title={The political blogosphere and the 2004 US election: divided they blog},
  author={Adamic, Lada A and Glance, Natalie},
  booktitle={Proceedings of the International Workshop on Link Discovery},
  pages={36--43},
  year={2005}
}

@inproceedings{zhang2021magnet,
  title={Magnet: {A} neural network for directed graphs},
  author={Zhang, Xitong and He, Yixuan and Brugnone, Nathan and Perlmutter, Michael and Hirn, Matthew},
  booktitle={Advances in Neural Information Processing Systems},
  volume={34},
  pages={27003--27015},
  year={2021}
}

@book{de1978practical,
  title={A practical guide to splines},
  author={De Boor, Carl},
  volume={27},
  year={1978},
  publisher={springer New York}
}

@article{kunoth2018foundations,
  title={Foundations of spline theory: B-splines, spline approximation, and hierarchical refinement},
  author={Kunoth, Angela and Lyche, Tom and Sangalli, Giancarlo and Serra-Capizzano, Stefano and Lyche, Tom and Manni, Carla and Speleers, Hendrik},
  journal={Splines and PDEs: From Approximation Theory to Numerical Linear Algebra: Cetraro, Italy 2017},
  pages={1--76},
  year={2018},
  publisher={Springer}
}

@book{natanson,
  title={Constructive function theory},
  author={Natanson, Isidor Pavlovich},
  year={1964},
  publisher={Ungar}
}

@book{freudbk,
  title={Orthogonal polynomials},
  author={Freud, G{\'e}za},
  year={2014},
  publisher={Elsevier}
}

@book{devlorbk,
  title={Constructive approximation},
  author={DeVore, R. A. and Lorentz, G. G.},
  volume={303},
  year={1993},
  publisher={Springer Science \& Business Media}
}

@article{jain2010data,
  title={Data clustering: 50 years beyond K-means},
  author={Jain, Anil K},
  journal={Pattern recognition letters},
  volume={31},
  number={8},
  pages={651--666},
  year={2010},
  publisher={Elsevier}
}

@article{chuidiamond90,
  title={A general framework for local interpolation},
  author={Chui, C. K. and Diamond, H.},
  journal={Numerische Mathematik},
  volume={58},
  number={1},
  pages={569--581},
  year={1990},
  publisher={Springer}
}

@article{quasiint2000,
  title={Quasi-interpolation in shift invariant spaces},
  author={Mhaskar, H. N. and Narcowich, F. J. and Ward, J. D.},
  journal={Journal of mathematical analysis and applications},
  volume={251},
  number={1},
  pages={356--363},
  year={2000},
  publisher={Elsevier}
}

@article{rey2023robust,
  title={Robust graph filter identification and graph denoising from signal observations},
  author={Rey, Samuel and Tenorio, Victor M and Marqu{\'e}s, Antonio G},
  journal={IEEE Transactions on Signal Processing},
  volume={71},
  pages={3651--3666},
  year={2023},
  publisher={IEEE}
}



\newpage

\section*{Supplementary Materials}
\subsection{Training Details}
\subsubsection{Code Realisation}
The detailed experimental code is available at 

\href{https://github.com/kellysylvia77/Digraph}{https://github.com/kellysylvia77/Digraph}.

\subsubsection{Datasets}
We evaluate the proposed method on seven real-world directed graphs with different scales and structural characteristics. As summarized in Table~\ref{Table:Datasets}, the datasets contain between $183$ and $5201$ vertices and between $298$ and $217073$ directed edges. We report the numbers of vertices, edges, and classes, together with the 
the homophily ratio $\mathcal{H}$. The five labeled datasets exhibit different degrees of heterophily, while Telegram and Blog contain no node labels; therefore, their class counts and homophily ratios are reported as unavailable.
\begin{table*}[htbp!]
\caption{Statistics of real-world directed graphs.}
\label{Table:Datasets}
\vspace*{1mm}
\centering
\begin{tabular}{cccccc|cc}
\toprule
\textbf{Datasets}  & \textbf{Cora} & \textbf{Squirrel} & \textbf{Cornell} & \textbf{Wisconsin} & \textbf{Texas} & \textbf{Telegram} & \textbf{Blog}    \\\midrule
\#Nodes, $|\mathcal{V}|$  & 2708  & 5201    & 183  & 251  & 183  & 245  & 5201 \\
\#Edges, $|\mathcal{E}|$  & 5429  & 217073  & 298  & 515  & 325  & 8912  & 19024  \\
\#Classes, $c$            & 7     & 5       & 5  & 5  & 5  & -  & -  \\
Hom. Ratio, $\mathcal{H}$ & 0.3347& 0.0854  & 0.1153  & 0.1325  & 0.0695  & N/A  & N/A  \\
\bottomrule 
\end{tabular}
\end{table*}
\begin{itemize}
    \item \textbf{Citation Network.}
    Cora is a citation network in which vertices represent scientific
    publications and directed edges denote citation relations. Node features are derived from document contents, and labels indicate research categories~\cite{sen2008collective}.

    \item \textbf{Wikipedia and Webpage Networks.}
    Squirrel is a Wikipedia webpage network whose directed edges represent hyperlinks between pages~\cite{rozemberczki2021multi}. Cornell, Wisconsin, and Texas are WebKB networks collected from university websites, where vertices correspond to webpages and directed edges represent hyperlinks~\cite{pei2020geom}. Their labels describe webpage categories.

    \item \textbf{Social and Information Networks.}
    Telegram is a directed influence network between Telegram channels. It is used to evaluate community discovery without supervision \cite{zhang2021magnet,he2022digrac}. Blog is a political blog network in which directed edges represent hyperlinks between blogs \cite{adamic2005political}. Both datasets are treated as unlabeled
    networks in our experiments.
\end{itemize}
The labeled datasets have homophily ratios ranging from $0.0695$ to
$0.3347$, indicating predominantly heterophilic connectivity. Together with the two unlabeled networks, these datasets provide diverse settings for evaluating both semi-supervised hierarchical clustering and unsupervised community discovery on directed graphs.

\subsubsection{Competitors}
We compare \myproj~with five representative spectral clustering methods for
directed graphs. These baselines cover several major strategies for handling
directionality, including matrix symmetrization, singular-vector
decomposition, and Hermitian/skew-symmetric spectral formulations.
\begin{itemize}
    \item \textbf{Bi-Sym}~\cite{satuluri2011symmetrizations} transforms the
    directed adjacency matrix into a symmetric similarity matrix through
    products involving $A$ and $A^\top$, so that vertices are compared
    according to shared predecessor and successor relationships before
    applying conventional spectral clustering.
    \item \textbf{DD-Sym}~\cite{satuluri2011symmetrizations} is a
    degree-discounted symmetrization method that further reweights the
    predecessor--successor similarities to reduce the dominance of
    high-degree vertices in the resulting spectral representation.
    \item \textbf{DI-SIM}~\cite{rohe2016co} constructs a regularized directed
    Laplacian and uses its left and right singular vectors to characterize
    the distinct sending and receiving roles of vertices, respectively.
    \item \textbf{Herm}~\cite{cucuringu2020hermitian} encodes edge
    orientation in a complex Hermitian matrix. Its real eigenvalues and
    orthonormal eigenvectors enable standard spectral embedding while
    preserving asymmetric connectivity information.
    \item \textbf{Skew}~\cite{hayashi2022skew} represents directed
    connectivity through a real skew-symmetric formulation closely related
    to the Hermitian construction, preserving the relevant directed-cut
    information while reducing the need for complex-domain computation.
\end{itemize}
All baselines are evaluated using the hyperparameter settings recommended in
their original implementations or publications. The same experimental
protocol is applied to the synthetic DSBM graphs and all seven real-world
directed networks to ensure a consistent comparison.

\subsubsection{Metrics}
We evaluate clustering performance using Modularity ($\mathcal{M}$), Adjusted Rand Index ($\mathcal{A}$), and F-measure ($\mathcal{F}$). Higher values indicate better clustering quality. The mathematical formulations of these metrics are provided in Definitions \ref{Modularity}-\ref{F-measure}.

\begin{definition}\label{Modularity}
{\rm [Modularity]} Let $k_i^{out}$ and $k_j^{in}$ denote the outdegree of node $i$ and the indegree of node $j$ in $W$, respectively. We assume that in a random directed graph with the same connectivity, the expected probability of an edge from $i$ to $j$ is $k_i^{out}k_j^{in} / m$. Here, $m=k_i^{in}+k_i^{out}$ is the total weight of the incoming/outgoing edges in the directed graph, which equals the sum of all entries in $W$. Then, the modularity metric is defined as:
\begin{equation}
\mathcal{M} = \frac{1}{m}\sum_{i,j}\left(W_{ij}-\frac{k_i^{out}k_j^{in}}{m}\right)\delta(C_i,C_j),
\end{equation}
where $\delta(C_i,C_j)$ is 1 if the nodes $i$ and $j$ both belong to the same cluster $C=C_i=C_j$, and 0 otherwise. This metric quantifies the difference between the actual number of edges within clusters and the expected number in a random graph with the same degree distribution.
\end{definition}

\begin{definition}\label{ARI}
{\rm [Adjusted Rand Index (ARI)]} Suppose we have two partitions for $N$ individuals: $C_1$ (with $K$ classes) and $C_2$ (with $M$ classes). Each individual $i$ has two labels: $c_i^1$ from $C_1$ and $c_i^2$ from $C_2$. We define two indicator variables, $c_{ij}^1$ and $c_{ij}^2$, to show if the pair $(i, j)$ is grouped together in $C_1$ and $C_2$, respectively:
\begin{equation}
c_{ij}^1 = 
\begin{cases} 
1 & \text{if } c_i^1 = c_j^1 = k, \\
0 & \text{otherwise},
\end{cases}
\quad \text{and} \quad
c_{ij}^2 = 
\begin{cases} 
1 & \text{if } c_i^2 = c_j^2 = \ell, \\
0 & \text{otherwise}.
\end{cases}
\end{equation}
The values $c_{ij}^1$ and $c_{ij}^2$ are realizations of Bernoulli random variables. A pair is consistent in similarity if $c_{ij}^1 c_{ij}^2 = 1$, and consistent in difference if $(1-c_{ij}^1)(1-c_{ij}^2) = 1$. The Adjusted Rand Index (ARI) is then defined over all pairs as:
\begin{equation}
\mathcal{A}\!\!=\!\!\frac{\sum_{ij} \binom{c_{ij}}{2}\!\!-\!\!\left[ \sum_i \binom{c_i}{2} \sum_j \binom{c_j}{2} \right] / \binom{n}{2}}{\frac{1}{2} \left[ \sum_i \binom{c_i}{2}\!\!+\!\!\sum_j \binom{c_j}{2} \right]\!\!-\!\!\left[ \sum_i \binom{c_i}{2} \sum_j \binom{c_j}{2} \right] / \binom{n}{2}}.
\end{equation}
\end{definition}

\begin{definition}\label{F-measure}
{\rm [F-measure Score]} Let ${C_1, C_2, \cdots, C_K}$ be the clusters obtained from a clustering algorithm, and let ${L_1, L_2, \cdots, L_M}$ be the ground-truth partition of nodes by class labels (i.e., $L_j$ contains all nodes in $W$ with label $j$). The metric is defined as:
\begin{equation}
F(C_i)=2\max_{1\leq j\leq M}\frac{|C_i\cap L_j|}{|C_i|+|L_j|},
\end{equation}
the (micro-averaged) F-measure is then defined by
\begin{equation}
\mathcal{F}=\frac{\sum_i|C_i|F(C_i)}{\sum_i|C_i|}.
\end{equation}
\end{definition}


\subsubsection{Directed Stochastic Block Model} 
The directed stochastic block model (DSBM) extends the classical stochastic block model (SBM) by incorporating directional interactions between clusters.
Besides the within-cluster and between-cluster connection probabilities $p$ and $q$, the DSBM introduces an orientation matrix $F\in[0,1]^{k\times k}$ to characterize the preferred direction of edges between different clusters. Hence, $F$ can be interpreted as the weighted adjacency matrix of a directed meta-graph whose vertices represent clusters.

Let $\mathcal{C}=\{C_1,\ldots,C_k\}$ be a partition of $N=k\cdot n$ vertices into $k$ clusters, each containing $n$ vertices. For a pair of vertices $u\in C_i$ and $v\in C_j$, an edge is generated with probability
\begin{equation}
\Pr(\{u,v\}\in E) =
\begin{cases}
p, & i=j,\\
q, & i\neq j.
\end{cases}
\end{equation}
Once an edge is generated, its orientation is determined by $F$. Specifically,
\begin{equation}
\Pr(u\rightarrow v\mid \{u,v\}\in E)=F_{ij},
\qquad
\Pr(v\rightarrow u\mid \{u,v\}\in E)=F_{ji},
\end{equation}
where
\begin{equation}
F_{ij}+F_{ji}=1, 
\qquad
F_{ii}=\frac{1}{2}.
\end{equation}
Thus, $F_{ij}>1/2$ indicates a preferred direction from cluster $C_i$ to cluster $C_j$, whereas $F_{ij}=1/2$ corresponds to no directional preference. In our synthetic settings, the off-diagonal entries of $F$ are parameterized by $\eta\in[0,\frac12]$, with the preferred and reverse directions assigned probabilities $1-\eta$ and $\eta$, respectively. Hence, a smaller $\eta$ indicates stronger directional preference, while $\eta=\frac12$ removes the directional bias. The resulting random digraph is denoted by $G(k,n,p,q,F)$.


\textbf{Example.}
Let $k=4$, $p=q$, and
\begin{equation}
F=
\begin{pmatrix}
\frac12 & \frac23 & \frac23 & \frac13\\
\frac13 & \frac12 & \frac23 & \frac23\\
\frac13 & \frac13 & \frac12 & \frac23\\
\frac23 & \frac13 & \frac13 & \frac12
\end{pmatrix}.
\label{eq:F_example}
\end{equation}

\begin{figure}[htpb!]
\centering
\includegraphics[width=0.25\linewidth]{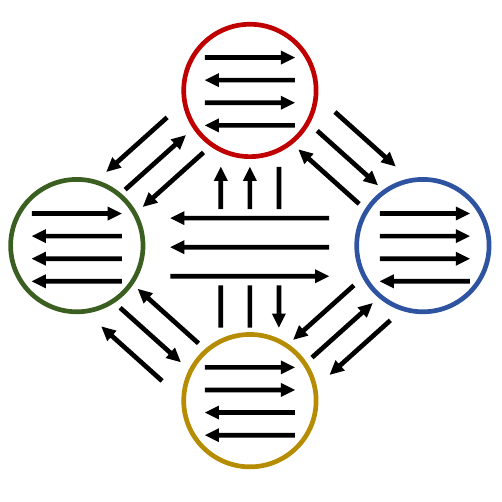}
\caption{Illustration of the asymmetric orientation pattern.}
\label{fig:F_example}
\end{figure}


In this case, $G$ consists of four equal-sized clusters $C_1,C_2,C_3$, and $C_4$, and any pair of vertices is connected with the same probability $p$. Hence, the cluster structure cannot be identified from edge density alone. Within each cluster, edge directions are chosen uniformly at random since $F_{ii}=1/2$. In contrast, the directions of edges between different clusters are determined non-uniformly by $F$. For every pair $(C_i,C_j)$ with $F_{ij}=2/3$, approximately two-thirds of the edges are expected to be directed from $C_i$ to $C_j$, while the remaining one-third are directed in the reverse direction. Specifically, the preferred inter-cluster directions are $C_1\rightarrow C_2$, $C_1\rightarrow C_3$, $C_2\rightarrow C_3$, $C_2\rightarrow C_4$, $C_3\rightarrow C_4$, and $C_4\rightarrow C_1$. Therefore, although $p=q$ eliminates density-based separation, the asymmetric directional pattern encoded by $F$ still provides informative cluster structure. This example illustrates how the DSBM can generate directed communities whose distinctions arise primarily from edge orientation rather than connectivity density.

\subsubsection{$k$-Means Clustering} 
Given $N$ observations $z_1, z_2,\ldots, z_N\in\mathbb R^n$ of $n$-dimensional vectors associated with vertices $v_1,\ldots, v_N$ of a vertex set $V$, $k$-means clustering \cite{jain2010data} aims to partition the $N$ observations into a $k$-way clustering $\mathcal C = \{C_1, C_2, \ldots, C_k\}$ of $V$, equivalently, to obtain a label signal $y:V\rightarrow \{1,\ldots,k\}$, so as to minimize the within-cluster sum of squares. In the semi-supervised setting, labels are known only for vertices indexed by $\mathcal{I}_\mathrm{knw}\subseteq\{1,\ldots,N\}$ so that  $Y_{\mathrm{knw}}(i)\in\{1,\ldots,k\}$ is given in advance for $i\in \mathcal{I}_\mathrm{knw}$. If $\mathcal{I}_{\mathrm{knw}}=\emptyset$, then it corresponds to the unsupervised setting. The assignments of labels for the remaining unknown vertices are inferred subject to the known labels.
Formally, the objective is to find:
\begin{equation}\label{k-means}
\min_{\mathcal C=\{C_j=y^{-1}(j)\}_{j=1}^k}\sum_{j=1}^{k}\sum_{v_i\in C_j}
\left\|z_i-c_j\right\|_2^2,\mbox{ s.t. } y|_{\mathcal I_{\mathrm{knw}}}=Y_{\mathrm{knw}},
\end{equation}
where each
\[
c_j=\frac{1}{|C_j|}\sum_{v_i\in C_j}z_i
\]
is the centroid of the vertices in the class $C_j$ for $j=1,\ldots, k$, the norm $\|\cdot\|_2$ is the Euclidean distance, but can be replaced by other general metric. 

The most common algorithm uses an iterative refinement technique. Given an initial set of $k$ centroids, the algorithm proceeds by alternating between two steps:
\begin{enumerate}
\item[1)] {\bf Assignment step}: Assign each observation to the cluster with the nearest mean (centroid). Mathematically, this means partitioning the observations according to the Voronoi diagram generated by the means.

\item[2)]{\bf Update step}: Recalculate means (centroids) for observations assigned to each cluster. This is also called refitting.
\end{enumerate}
After each iteration, the within-cluster sum of squares decreases monotonically, yielding a nonnegative, monotonically decreasing sequence. This guarantees that the $k$-means always converges, but not necessarily to the global optimum. Algorithm~\ref{k-means} presents the details of the implementations of $k$-means in our setting.

\begin{algorithm}[htpb!]
\caption{$k$-means Clustering Algorithm.}
\begin{algorithmic}[1]
\item[{\rm a)}] {\bf Input}: $V=\{v_1,\ldots,v_n\}$ with each vertex $v$ associated with a $d$-dimensional vector $z_v\in\mathbb R^n$, respectively; the known pair $(\mathcal{I}_{\mathrm{knw}},Y_{\mathrm{knw}})$ of labels; number $k$ of clusters; the maximum number of iterations
\item[{\rm b)}] {\bf Output}: A $k$-way clustering $\mathcal C=\{C_1,\ldots, C_k\}$ of $V$
\item[{\rm c)}] {\bf Main Steps}:
\STATE Initialization: Randomly choose $k$ vertices $u_1,\ldots,u_k$ from $V$ as centers subject to  $u_j\in Y_{\mathrm{knw}}^{-1}(j)$, respectively. The initial centroids $c_1,\ldots, c_k$ are then the vectors $z_{u_1},\ldots, z_{u_k}$, respectively
\WHILE{true}
\STATE Construct cluster $C_j$ for $j=1,\ldots,k$. That is, $v\in V$ belongs to $C_j$ if either $v\in Y_{\mathrm{knw}}^{-1}(j)$ or $j = \mathrm{argmin}_{1\le j'\le k}  \|z_v-c_{j'}\|_2$
\STATE Update the centers:  for each $C_j$, find a new center $u\in C_j$ such that $\sum_{v\in C_j}\|z_u-z_v\|_2$ is minimal and the centroid is updated to $z_u$, respectively
\STATE Break if all centers remain the same or the maximum number of iterations is reached.
\ENDWHILE
\end{algorithmic}
\label{k-means}
\end{algorithm}

\subsection{Supplementary Experiments}
\subsubsection{Results for the DSBM}
This section provides a further analysis of model performance using an expanded variable set in DSBM including $N$, $k$ and the pattern of $F$ matrix, as well as under more conditions than those covered in the main text.

\begin{figure*}[h]
\setlength{\abovecaptionskip}{0.cm}
\setlength{\belowcaptionskip}{-0.cm}
\centering
\includegraphics[width=0.99\textwidth]{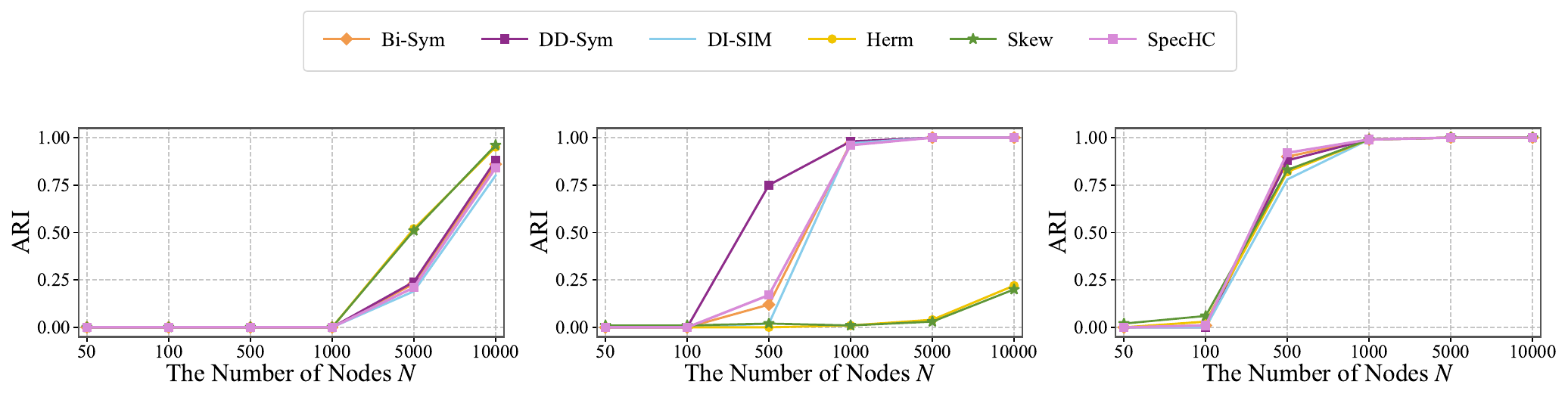}
\caption{The models' performance is compared across varying numbers of nodes under three scenarios: $p = q$ (left), $p \gg q$ (middle), and $q \gg p$ (right).}\label{fig:n_pq}
\end{figure*}

\begin{figure*}[h]
\setlength{\abovecaptionskip}{0.cm}
\setlength{\belowcaptionskip}{-0.cm}
\centering
\includegraphics[width=0.99\textwidth]{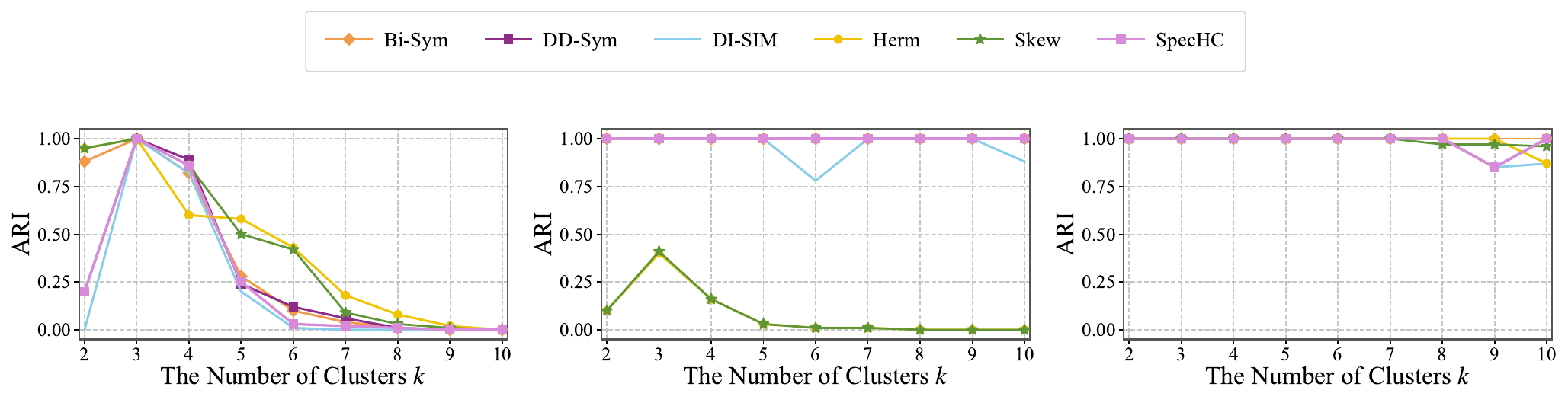}
\caption{The models' performance is compared across varying numbers of clusters under three scenarios: $p = q$ (left), $p \gg q$ (middle), and $q \gg p$ (right).}\label{fig:k_pq}
\end{figure*}

\noindent\textbf{Parameter $N$--Graph Scale. }
Figure~\ref{fig:n_pq} examines the effect of graph size on clustering
performance under three structural settings: $p=q$, $p\gg q$, and
$q\gg p$. Overall, increasing $N$ provides more structural evidence and
generally improves the ARI of all methods. When $p=q$, the cluster structure is relatively ambiguous, and most methods remain ineffective on small graphs. Their performance improves only when the graph becomes sufficiently large, with ARI values rising markedly at $N=5000$ and $N=10000$. Clearer structural separation leads to faster improvement. Under $q\gg p$, all methods achieve high ARI at $N=500$ and nearly perfect clustering once $N\geq1000$. Under $p\gg q$, however, the methods respond differently: DD-Sym performs well from $N=500$, while Bi-Sym, DI-SIM, and \myproj~reach near-perfect performance at approximately $N=1000$. In contrast, Herm and Skew remain less effective even on larger graphs in this setting. These results show that graph scale and structural pattern jointly determine clustering difficulty. Larger graphs generally improve recoverability, while strong homophilic or heterophilic separation reduces the amount of data
required. \myproj~exhibits stable scalability and achieves nearly perfect clustering on sufficiently large graphs with clearly distinguishable structures.

\begin{figure*}[t]
\setlength{\abovecaptionskip}{0.cm}
\setlength{\belowcaptionskip}{-0.cm}
\centering
\subfigure[Circular Pattern]{\includegraphics[width=0.99\textwidth]{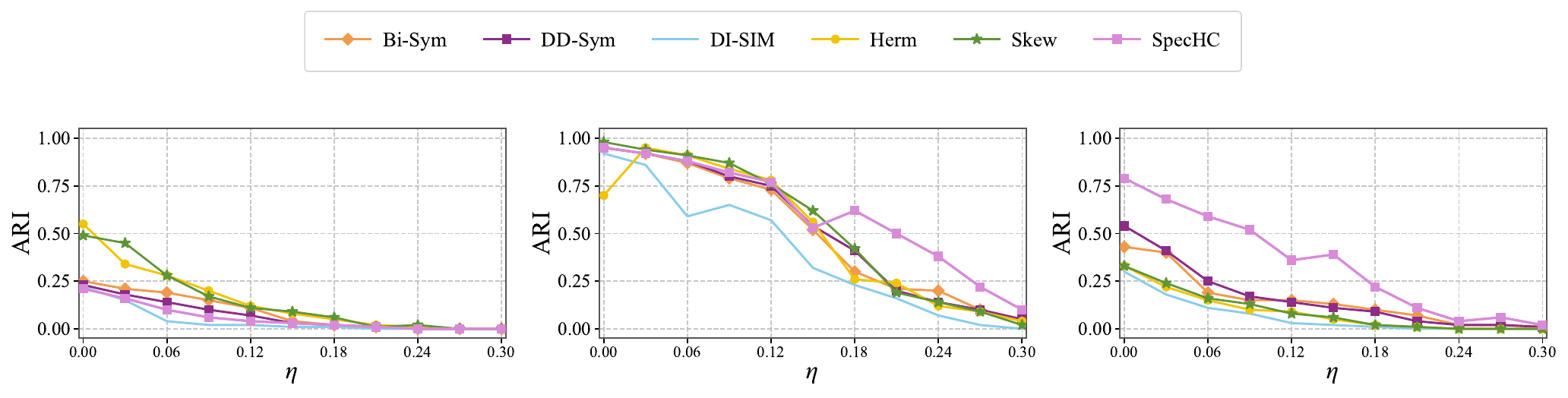}}
\subfigure[Complete Meta-Graph]{\includegraphics[width=0.99\textwidth]{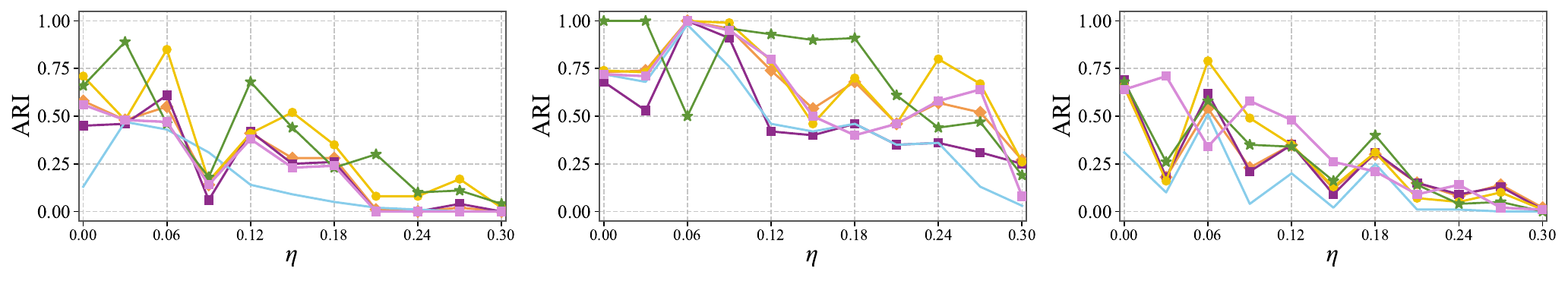}}
\caption{Performance comparison of the models across varying levels of $\eta$ under two distinct $F$ matrix patterns: (a) Circular Pattern and (b) Complete Meta-Graph.}\label{fig:F_pattern}
\end{figure*}

\begin{figure*}[t!]
\setlength{\abovecaptionskip}{0.cm}
\setlength{\belowcaptionskip}{-0.cm}
\centering
\subfigure[$\eta=0$, $p=0.0045$, $q=0.0045$]{\includegraphics[width=0.32\textwidth]{alpha_beta_eta00_p045_q045.pdf}}
\subfigure[$\eta=0.12$, $p=0.0045$, $q=0.0045$]{\includegraphics[width=0.32\textwidth]{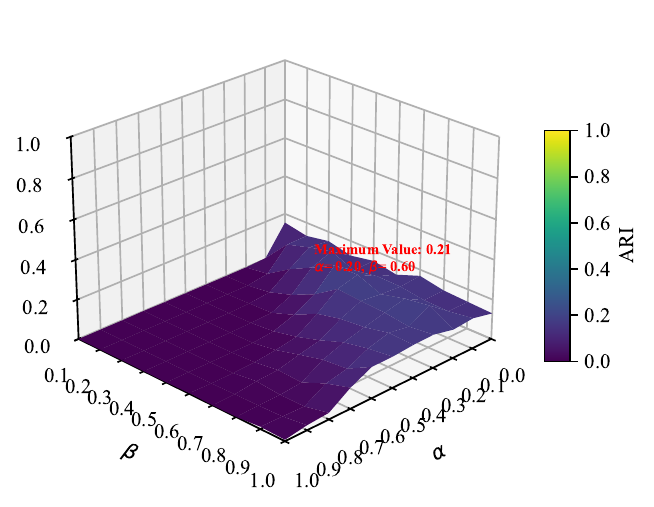}}
\subfigure[$\eta=0.24$, $p=0.0045$, $q=0.0045$]{\includegraphics[width=0.32\textwidth]{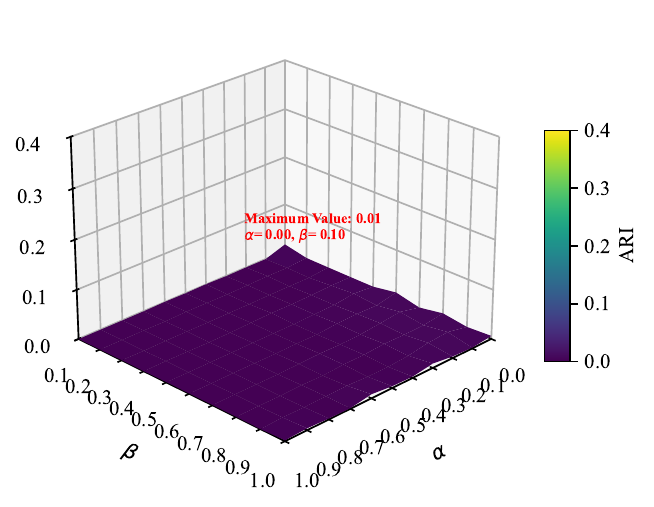}}
\subfigure[$\eta=0$, $p=0.0045$, $q=0.05$]{\includegraphics[width=0.32\textwidth]{alpha_beta_eta00_p045_q50.pdf}}
\subfigure[$\eta=0.12$, $p=0.0045$, $q=0.05$]{\includegraphics[width=0.32\textwidth]{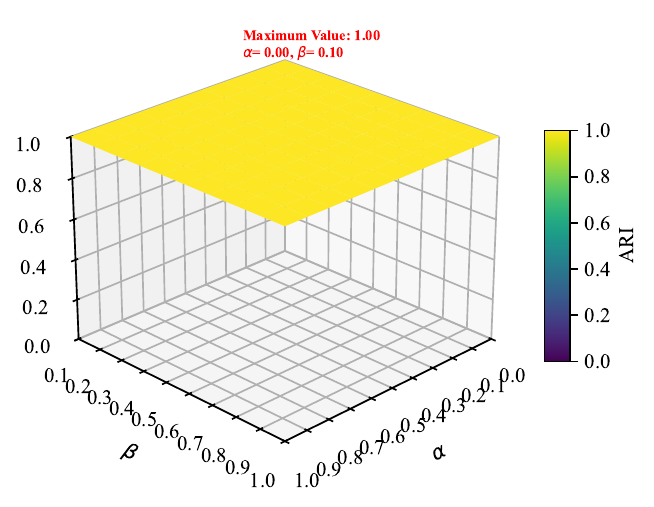}}
\subfigure[$\eta=0.24$, $p=0.0045$, $q=0.05$]{\includegraphics[width=0.32\textwidth]{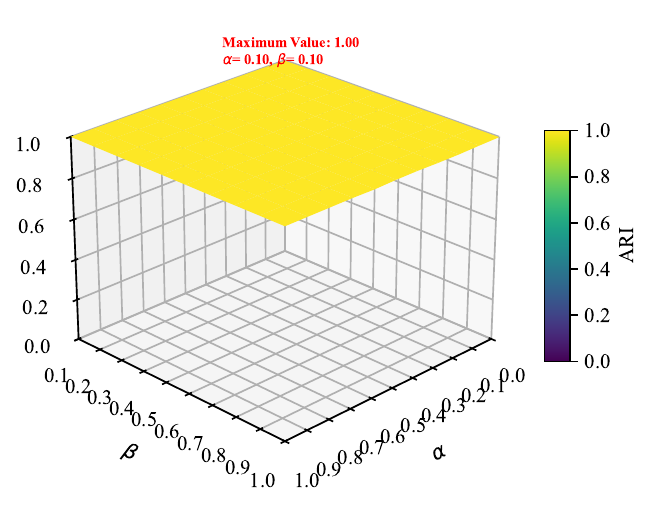}}
\subfigure[$\eta=0$, $p=0.05$, $q=0.0045$]{\includegraphics[width=0.32\textwidth]{alpha_beta_eta00_p50_q045.pdf}}
\subfigure[$\eta=0.12$, $p=0.05$, $q=0.0045$]{\includegraphics[width=0.32\textwidth]{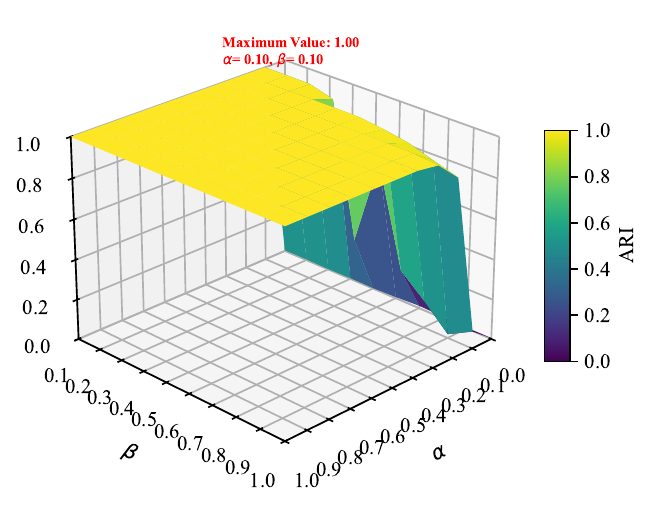}}
\subfigure[$\eta=0.24$, $p=0.05$, $q=0.0045$]{\includegraphics[width=0.32\textwidth]{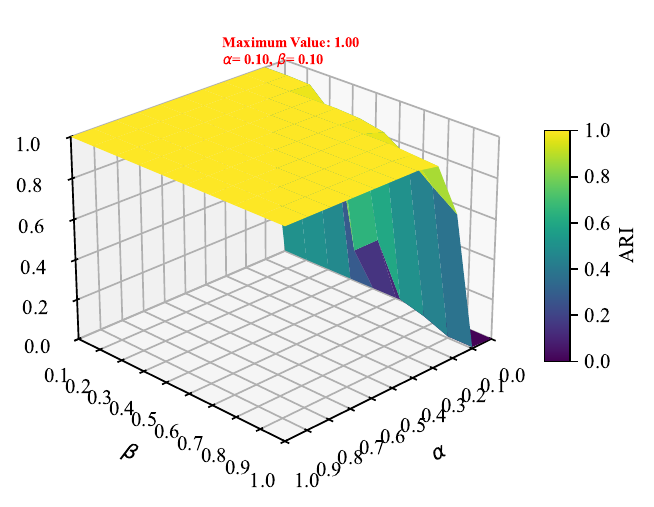}}
\caption{Investigating the effect of $\alpha$ and $\beta$: Model parameter sensitivity on different directed graph structures.}\label{fig:alpha_beta}
\end{figure*}

\noindent\textbf{Parameter $k$--Cluster Size. }
Figure~\ref{fig:k_pq} examines the effect of the number of clusters under three structural settings: $p=q$, $p\gg q$, and $q\gg p$. When $p=q$, the community structure is weak, and the performance of all methods decreases rapidly as $k$ increases. Most ARI values approach zero for large $k$, indicating that finer partitions are difficult to recover without clear structural separation. When $p\gg q$, \myproj, Bi-Sym, and DD-Sym remain nearly perfect across the tested values of $k$, while DI-SIM exhibits only a small fluctuation. In contrast, Herm and Skew perform poorly in this setting. Under $q\gg p$, almost all methods maintain high ARI, although slight degradation appears for some methods when $k$ becomes large. Overall, increasing the number of clusters amplifies clustering difficulty mainly when the underlying block structure is ambiguous. Clear homophilic or heterophilic patterns substantially reduce this sensitivity, and \myproj~shows strong robustness across different clustering granularities.

\noindent\textbf{Pattern Type for $F$ Matrix. }
The parameter $F \in [0,1]^{k \times k}$ dictates the edge direction probability matrix between clusters, satisfying $F_{i,j} + F_{j,i} = 1$. Figure \ref{fig:F_pattern} investigates the influence of this meta-graph structure on clustering robustness against the directional noise parameter $\eta$. We explicitly evaluate two distinct structural definitions:

(1) Cyclic Pattern: Directional biases are strictly confined to adjacent clusters forming a directed cycle (e.g., $F_{i, i+1} = 1-\eta$ and $F_{i+1, i} = \eta$), while all other non-adjacent cluster pairs maintain perfectly symmetric, unbiased connections ($F_{i,j} = 1/2$).

(2) Complete Pattern: Every distinct pair of clusters possesses a strict directional bias ($F_{i,j} \in \{\eta, 1-\eta\}$ for all $i \neq j$), forming a densely connected and highly complex asymmetric meta-graph.

To provide a comprehensive evaluation, we expand the analysis across three distinct structural scenarios: $p=q$ (left), $q > p$ (middle), and $p > q$ (right). The results reveal a stark contrast between the two topologies. Under the Cyclic Pattern (Figure \ref{fig:F_pattern}(a)), our proposed methods exhibit extraordinary tolerance to directional noise whenever explicit structural signals exist ($p \neq q$). Specifically, in the heterophilic setting ($q \gg p$), \myproj~degrades at a significantly slower rate than baselines. In the homophilic setting ($p \gg q$), our methods establish and maintain a commanding lead from the outset. Conversely, under the Complete Pattern (Figure \ref{fig:F_pattern}(b)), the densely connected and conflicting directional biases make the extraction of stable signals exceptionally difficult. This leads to severe performance oscillations across all algorithms, regardless of the relationship between $p$ and $q$. Ultimately, these findings highlight that our methods are exceptionally adept at capturing sparsely structured meta-graph signals (such as cyclic patterns) and offer outstanding resilience against structural noise.

\noindent\textbf{Parameter $\alpha$ and $\beta$–Spectral Balance. }
Figure~\ref{fig:alpha_beta} examines the sensitivity of \myproj~to
$\alpha$ and $\beta$, which control the symmetric connectivity term and the asymmetric directional term in the generalized Hermitian matrix,
respectively. The results show that their appropriate balance depends strongly on the underlying graph structure. When $p=q$, the graph provides little separation through edge density. At $\eta=0$, useful clustering is obtained only within a limited parameter region, and the overall ARI decreases rapidly as $\eta$ increases. This indicates that, without clear homophilic or heterophilic structure, the model relies heavily on informative edge directions and is therefore sensitive to directional perturbations. When $q\gg p$, a stronger contribution from the asymmetric component is generally beneficial at $\eta=0$. For $\eta=0.12$ and $0.24$, however, the ARI remains close to $1$ over almost the entire parameter space, suggesting that the pronounced inter-cluster connectivity makes the clustering result largely insensitive to the precise choice of $\alpha$ and $\beta$. In contrast, when $p\gg q$, high ARI is obtained over a broad region where the symmetric connectivity component is sufficiently emphasized. Performance drops sharply when the asymmetric term dominates excessively, showing that
within-cluster connectivity is the primary source of information in strongly homophilic graphs. Overall, no single parameter pair is optimal for all settings: $\beta$ is more important for direction-dominated structures, whereas $\alpha$ should receive greater weight when homophilic connectivity provides the main clustering signal.

\begin{table*}[t]
\caption{Clustering algorithms on the Cora dataset: a comparative analysis of modularity ($\mathcal{M}$), and F-measure ($\mathcal{F}$). The best-performing model is highlighted in \textbf{bold}, and the second-best is marked with an \underline{underline}.}\label{Table:Cora}
\vspace*{1mm}
\centering
\scalebox{0.85}[0.85]{\begin{tabular}{cccccccccccc}
\toprule
$\mathcal{M}$& \textbf{Trains (\%)}  & \textbf{0 (USL)}  & \textbf{10}  & \textbf{20}  & \textbf{30}  & \textbf{40}  & \textbf{50}  & \textbf{60}  & \textbf{70}  & \textbf{80}  & \textbf{90}  \\ \midrule
& \textbf{Level 2 (70)} & 0.4392 & 0.3759 & 0.3760 & 0.4318 & 0.4313 & 0.4651 & 0.5622 & 0.5964 & \underline{0.6482} & 0.7083  \\
\multirow{-2}{*}{\textbf{Bi-Sym}} & \textbf{Level 1 (7)} & 0.2624 & 0.2546 & 0.2233 & 0.2625 & 0.1476 & 0.2048 & 0.2752 & 0.3033 & 0.0907 & 0.0309  \\
& \textbf{Level 2 (70)} & \underline{0.4983} & 0.4089 & 0.3501 & 0.4219 & 0.4076 & 0.4667 & 0.5405 & 0.5944 & 0.6365 & 0.7078  \\
\multirow{-2}{*}{\textbf{DD-Sym}} & \textbf{Level 1 (7)} & 0.4789 & 0.3494 & 0.3652 & 0.3888 & 0.4120 & 0.3817 & 0.2040 & 0.4070 & 0.2990 & 0.1125  \\
& \textbf{Level 2 (70)} & \textbf{0.6091} & \textbf{0.4782} & \textbf{0.4076} & \textbf{0.4730} & \textbf{0.4436} & \underline{0.4679} & 0.5671 & 0.6081 & 0.6464 & 0.7089  \\
\multirow{-2}{*}{\textbf{DI-SIM}} & \textbf{Level 1 (7)} & 0.4456 & \underline{0.4293} & 0.3047 & 0.2351 & 0.2629 & 0.2622 & 0.2127 & 0.2855 & 0.2282 & 0.1297  \\
& \textbf{Level 2 (70)} & 0.4091 & 0.3492 & 0.3591 & 0.4286 & 0.4020 & 0.4671 & \underline{0.5674} & \underline{0.6101} & 0.6316 & \underline{0.7116}  \\
\multirow{-2}{*}{\textbf{Herm}} & \textbf{Level 1 (7)} & 0.1376 & 0.1400 & 0.1206 & 0.1575 & 0.1017 & 0.0698 & 0.1162 & 0.1109 & 0.0600 & 0.0680  \\
& \textbf{Level 2 (70)} & 0.4091 & 0.3492 & 0.3591 & 0.4286 & 0.4020 & 0.4671 & \underline{0.5674} & \underline{0.6101} & 0.6316 & \underline{0.7116}  \\
\multirow{-2}{*}{\textbf{Skew}} & \textbf{Level 1 (7)} & 0.1876 & 0.1485 & 0.1383 & 0.1810 & 0.0831 & 0.0685 & 0.0561 & 0.0814 & 0.0270 & 0.0646  \\
\midrule
& \textbf{Level 2 (70)} & 0.4565 & 0.3765 & 0.3756 & \underline{0.4429} & \underline{0.4370} & \textbf{0.4705} & \textbf{0.5716} & \textbf{0.6173} & \textbf{0.6492} & \textbf{0.7157}  \\
\multirow{-2}{*}{\textbf{\myproj}} & \textbf{Level 1 (7)} & 0.2509 & 0.2903 & \underline{0.3888} & 0.3576 & 0.2601 & 0.3819 & 0.3035 & 0.2071 & 0.0800 & 0.0544  \\
\toprule
$\mathcal{F}$& \textbf{Trains (\%)}  & \textbf{0 (USL)}  & \textbf{10}  & \textbf{20}  & \textbf{30}  & \textbf{40}  & \textbf{50}  & \textbf{60}  & \textbf{70}  & \textbf{80}  & \textbf{90}  \\ \midrule
& \textbf{Level 2 (70)} & 0.1211 & 0.1536 & 0.2474 & 0.3110 & 0.3929 & 0.4174 & 0.5801 & 0.5845 & 0.7847 & 0.8560  \\
\multirow{-2}{*}{\textbf{Bi-Sym}} & \textbf{Level 1 (7)} & 0.1556 & 0.2360 & 0.2885 & 0.3418 & 0.4360 & \textbf{0.5559} & 0.6190 & \underline{0.7354} & \underline{0.8266} & 0.9079  \\
& \textbf{Level 2 (70)} & 0.0915 & 0.0455 & 0.0834 & 0.1664 & 0.2456 & 0.4153 & 0.4589 & 0.5883 & 0.7739 & 0.8562  \\
\multirow{-2}{*}{\textbf{DD-Sym}} & \textbf{Level 1 (7)} & \textbf{0.4154} & \underline{0.2590} & 0.3098 & \underline{0.3851} & 0.4176 & 0.4993 & \textbf{0.6563} & \textbf{0.7614} & 0.8188 & \textbf{0.9152}  \\
& \textbf{Level 2 (70)} & \underline{0.2573} & 0.0317 & 0.1989 & 0.2470 & 0.3834 & 0.3438 & 0.5714 & 0.6687 & 0.7835 & 0.8790  \\
\multirow{-2}{*}{\textbf{DI-SIM}} & \textbf{Level 1 (7)} & 0.0959 & 0.1212 & 0.3144 & \textbf{0.3917} & \textbf{0.4736} & 0.4858 & \underline{0.6320} & 0.7100 & 0.8014 & 0.9071  \\
& \textbf{Level 2 (70)} & 0.1838 & 0.0254 & 0.2160 & 0.3293 & 0.2469 & 0.4139 & 0.5601 & 0.6830 & 0.7170 & 0.8822  \\
\multirow{-2}{*}{\textbf{Herm}} & \textbf{Level 1 (7)} & 0.1919 & 0.1904 & 0.2601 & 0.3753 & 0.4623 & 0.5239 & 0.6287 & 0.7239 & 0.8151 & 0.9077  \\
& \textbf{Level 2 (70)} & 0.1838 & 0.0254 & 0.2160 & 0.3293 & 0.2469 & 0.4139 & 0.5601 & 0.6830 & 0.7170 & 0.8822  \\
\multirow{-2}{*}{\textbf{Skew}} & \textbf{Level 1 (7)} & 0.1972 & 0.1723 & \underline{0.3473} & 0.3846 & \underline{0.4681} & \underline{0.5291} & 0.5944 & 0.7080 & \textbf{0.8590} & 0.9019  \\
\midrule
& \textbf{Level 2 (70)} & 0.1848 & 0.0245 & 0.0837 & 0.3194 & 0.2456 & 0.4176 & 0.5866 & 0.6853 & 0.7876 & 0.8874  \\
\multirow{-2}{*}{\textbf{\myproj}} & \textbf{Level 1 (7)} & 0.1592 & \textbf{0.3197} & \textbf{0.4129} & 0.3288 & \textbf{0.5420} & 0.5200 & \underline{0.6466} & 0.7096 & 0.8114 & \underline{0.9112} \\
\bottomrule
\end{tabular}}
\end{table*}

\begin{table*}[thpb!]
\caption{Clustering algorithms on the Squirrel dataset: a comparative analysis of modularity ($\mathcal{M}$), and F-measure ($\mathcal{F}$). The best-performing model is highlighted in \textbf{bold}, and the second-best is marked with an \underline{underline}.}\label{Table:Squirrel}
\vspace*{1mm}
\centering
\scalebox{0.85}[0.85]{\begin{tabular}{cccccccccccc}
\toprule
$\mathcal{M}$& \textbf{Trains (\%)}  & \textbf{0 (USL)}  & \textbf{10}  & \textbf{20}  & \textbf{30}  & \textbf{40}  & \textbf{50}  & \textbf{60}  & \textbf{70}  & \textbf{80}  & \textbf{90}  \\ \midrule
& \textbf{Level 2 (50)} & 0.1110 & 0.1043 & 0.0942 & 0.0934 & 0.0734 & 0.0730 & 0.0457 & 0.0426 & 0.0474 & 0.0518  \\
\multirow{-2}{*}{\textbf{Bi-Sym}} & \textbf{Level 1 (5)} & 0.1365 & 0.1379 & 0.1328 & \textbf{0.2097} & 0.1385 & \textbf{0.2942} & \underline{0.1899} & 0.1746 & \textbf{0.1866} & \textbf{0.2097}  \\
& \textbf{Level 2 (50)} & \underline{0.2289} & \underline{0.2137} & \underline{0.1929} & \underline{0.2062} & 0.1400 & 0.1280 & 0.1172 & 0.0756 & 0.0606 & 0.0792  \\
\multirow{-2}{*}{\textbf{DD-Sym}} & \textbf{Level 1 (5)} & \textbf{0.4736} & \textbf{0.2888} & \textbf{0.2644} & 0.1445 & \underline{0.2056} & 0.1856 & \textbf{0.1916} & \underline{0.1852} & 0.1686 & 0.0962  \\
& \textbf{Level 2 (50)} & 0.1743 & 0.1643 & 0.1409 & 0.1237 & \textbf{0.2747} & 0.0844 & 0.0753 & 0.0672 & 0.0511 & 0.0587  \\
\multirow{-2}{*}{\textbf{DI-SIM}} & \textbf{Level 1 (5)} & 0.2091 & 0.2107 & 0.1437 & 0.1579 & 0.1681 & \underline{0.1937} & 0.1546 & 0.1504 & 0.1160 & 0.0653  \\
& \textbf{Level 2 (50)} & 0.1103 & 0.0843 & 0.0955 & 0.0938 & 0.0792 & 0.0571 & 0.0539 & 0.0416 & 0.0400 & 0.0554  \\
\multirow{-2}{*}{\textbf{Herm}} & \textbf{Level 1 (5)} & 0.1756 & 0.1514 & 0.1778 & 0.1650 & 0.1404 & 0.1305 & 0.1247 & 0.1055 & 0.1607 & 0.1530  \\
& \textbf{Level 2 (50)} & 0.1103 & 0.0843 & 0.0955 & 0.0938 & 0.0792 & 0.0571 & 0.0539 & 0.0416 & 0.0400 & 0.0554  \\
\multirow{-2}{*}{\textbf{Skew}} & \textbf{Level 1 (5)} & 0.1604 & 0.1344 & 0.1907 & 0.1896 & 0.2049 & 0.0886 & 0.1031 & 0.0905 & \underline{0.1826} & 0.1342  \\ \midrule
& \textbf{Level 2 (50)} & 0.1201 & 0.1155 & 0.0975 & 0.0922 & 0.1260 & 0.0629 & 0.0487 & 0.0450 & 0.0507 & 0.0572  \\
\multirow{-2}{*}{\textbf{\myproj}} & \textbf{Level 1 (5)} & 0.1946 & 0.1856 & \underline{0.2210} & \textbf{0.2303} & \textbf{0.2758} & \underline{0.2610} & 0.1487 & \textbf{0.2074} & \textbf{0.2087} & \textbf{0.2453}  \\\toprule
$\mathcal{F}$& \textbf{Trains (\%)}  & \textbf{0 (USL)}  & \textbf{10}  & \textbf{20}  & \textbf{30}  & \textbf{40}  & \textbf{50}  & \textbf{60}  & \textbf{70}  & \textbf{80}  & \textbf{90}  \\ \midrule
& \textbf{Level 2 (50)} & 0.1412 & 0.1650 & 0.2223 & 0.3045 & 0.2721 & 0.5132 & 0.4807 & 0.6036 & 0.7883 & 0.8572  \\
\multirow{-2}{*}{\textbf{Bi-Sym}} & \textbf{Level 1 (5)} & 0.2284 & 0.2684 & \underline{0.3795} & 0.4330 & 0.5250 & \underline{0.5860} & \underline{0.6754} & 0.7455 & \underline{0.8337} & 0.9168  \\
& \textbf{Level 2 (50)} & 0.0028 & 0.0782 & 0.0963 & 0.1939 & 0.2572 & 0.3957 & 0.5504 & 0.6601 & 0.7330 & 0.8647  \\
\multirow{-2}{*}{\textbf{DD-Sym}} & \textbf{Level 1 (5)} & 0.2636 & 0.2829 & 0.3621 & 0.4204 & 0.4629 & 0.5794 & 0.6637 & \textbf{0.7494} & \textbf{0.8369} & \textbf{0.9208}  \\
& \textbf{Level 2 (50)} & 0.0023 & 0.2225 & 0.0857 & 0.1602 & 0.3963 & 0.3535 & 0.5935 & 0.5874 & 0.7202 & 0.8552  \\
\multirow{-2}{*}{\textbf{DI-SIM}} & \textbf{Level 1 (5)} & 0.2664 & \textbf{0.3221} & 0.3627 & \textbf{0.4492} & 0.5065 & 0.5697 & 0.6642 & 0.7447 & 0.8303 & \underline{0.9196}  \\
& \textbf{Level 2 (50)} & 0.1002 & 0.1841 & 0.0841 & 0.1488 & 0.4032 & 0.4830 & 0.5842 & 0.5970 & 0.7214 & 0.8966  \\
\multirow{-2}{*}{\textbf{Herm}} & \textbf{Level 1 (5)} & \textbf{0.3031} & \underline{0.3031} & 0.3696 & 0.3925 & \underline{0.5291} & \textbf{0.5999} & 0.6632 & \underline{0.7479} & 0.8314 & 0.9176  \\
& \textbf{Level 2 (50)} & 0.1002 & 0.1841 & 0.0841 & 0.1488 & 0.4032 & 0.4830 & 0.5842 & 0.5970 & 0.7214 & 0.8966  \\
\multirow{-2}{*}{\textbf{Skew}} & \textbf{Level 1 (5)} & 0.2526 & 0.3020 & \textbf{0.4047} & \underline{0.4347} & \textbf{0.5340} & 0.5842 & \textbf{0.6794} & 0.7437 & 0.8301 & 0.9190  \\
\midrule
& \textbf{Level 2 (50)} & 0.1250 & 0.2048 & 0.1107 & 0.3206 & 0.2559 & 0.4872 & 0.4605 & 0.5896 & 0.7933 & 0.8572  \\
\multirow{-2}{*}{\textbf{\myproj}} & \textbf{Level 1 (5)} & \underline{0.2946} & 0.2653 & 0.3667 & \underline{0.4471} & 0.4982 & \underline{0.5887} & 0.6591 & \underline{0.7479} & 0.8310 & \underline{0.9201}  \\
\bottomrule
\end{tabular}}
\end{table*}

\subsubsection{Results for Real-World Data}
Tables~\ref{Table:Cora}--\ref{Table:Texas} report the clustering results on Cora, Squirrel, Cornell, and Texas under different supervision ratios and hierarchical resolutions. Level~1 contains the same number of clusters as the ground-truth classes, whereas Level~2 provides a finer partition. We separately analyze Modularity ($\mathcal{M}$) and F-measure ($\mathcal{F}$), since they measure structural cohesion and class consistency, respectively.

\noindent\textbf{Structural Quality.}
The modularity results vary noticeably across datasets and hierarchical levels. On Cora, \myproj~is particularly effective at Level~2 under moderate-to-high supervision. It achieves the highest $\mathcal{M}$ from $50\%$ to $90\%$ labeled vertices, increasing from $0.4705$ to $0.7157$. This indicates that its finer partition captures structurally cohesive subcommunities when sufficient label constraints are available. Its Level~1 modularity is less competitive, suggesting that the coarse class partition does not always coincide with the densest graph structure. On Squirrel, the advantage of \myproj~is mainly observed at Level~1. It achieves the best modularity at the $30\%$, $40\%$, $70\%$, $80\%$, and $90\%$ settings, and remains competitive at several other ratios. In particular, its Level~1 modularity reaches $0.2453$ at $90\%$, whereas the Level~2 scores of all methods remain relatively low. This suggests that the coarse partition provides a more meaningful structural description of this strongly heterophilic network than the finer partition. The modularity results on Cornell and Texas are more irregular because these networks are small and their structural communities may not align closely with the class labels. On Cornell, \myproj~obtains the best Level~1 modularity at $10\%$, $30\%$, and $80\%$, while remaining competitive at several other settings. On Texas, it achieves the best result at Level~2 with $10\%$ supervision and at Level~1 with $80\%$ supervision. Although \myproj~does not dominate every setting on these two datasets, the results show that it can identify structurally cohesive partitions at selected resolutions and supervision levels.

\noindent\textbf{Class Consistency.}
Compared with modularity, the F-measure generally benefits more consistently from increasing supervision, especially at Level~1, where the number of clusters matches the number of classes. On Cora, \myproj~achieves the best Level~1 $\mathcal{F}$ at the $10\%$, $20\%$, and $40\%$ settings and the second-best result at $60\%$ and $90\%$. It also attains the best Level~2 score at $50\%$. These results demonstrate that \myproj~can effectively use limited labels to recover class-aligned partitions, particularly under low-to-moderate supervision. On Squirrel, the differences among the leading methods become relatively small as the supervision ratio increases. Nevertheless, \myproj~remains competitive across both levels and reaches a Level~1 F-measure of $0.9201$ at $90\%$, which is close to the best result of $0.9208$. Its strong results at several intermediate ratios further indicate that the hierarchical constraints remain effective despite the highly heterophilic structure of the network. On Cornell, \myproj~performs strongly in both unsupervised and supervised settings. It achieves the best F-measure at Level~2 under $0\%$ and $30\%$ supervision and at Level~1 under $70\%$ and $80\%$ supervision. At $90\%$, its Level~1 score of $0.9039$ is the second best. This shows that the method can capture both fine structural groups and coarse class-aligned clusters, depending on the available supervision. The clearest F-measure advantage is observed on Texas. At Level~1, \myproj~achieves the best result at $10\%$, $40\%$, $50\%$, $60\%$, $70\%$, and $90\%$ supervision. It also obtains the best Level~2 score at $70\%$. Its Level~1 F-measure reaches $0.9179$ at $90\%$, outperforming all baselines. These results confirm that the inherited label constraints are particularly effective for recovering the semantic class structure of Texas.

Overall, the two metrics reveal complementary properties of the proposed hierarchy. The modularity advantage is more dataset- and resolution-dependent, reflecting differences between structural communities and semantic classes. By contrast, \myproj~shows more consistent improvements in F-measure, especially at Level~1, demonstrating its effectiveness in incorporating limited supervision for class-consistent clustering.

\begin{table*}[t]
\caption{Clustering algorithms on the Cornell dataset: a comparative analysis of modularity ($\mathcal{M}$), and F-measure ($\mathcal{F}$). The best-performing model is highlighted in \textbf{bold}, and the second-best is marked with an \underline{underline}.}\label{Table:Cornell}
\vspace*{1mm}
\centering
\scalebox{0.85}[0.85]{\begin{tabular}{cccccccccccc}
\toprule
$\mathcal{M}$& \textbf{Trains (\%)}  & \textbf{0 (USL)}  & \textbf{10}  & \textbf{20}  & \textbf{30}  & \textbf{40}  & \textbf{50}  & \textbf{60}  & \textbf{70}  & \textbf{80}  & \textbf{90}  \\ \midrule
& \textbf{Level 2 (10)} & 0.1615 & 0.1126 & 0.1359 & 0.0933 & 0.0814 & 0.0835 & 0.1410 & 0.0724 & 0.1280 & 0.1811  \\
\multirow{-2}{*}{\textbf{Bi-Sym}} & \textbf{Level 1 (5)} & 0.0781 & 0.1328 & 0.0589 & 0.0954 & 0.1048 & 0.1523 & 0.0246 & \underline{0.1480} & 0.0187 & 0.1809  \\
& \textbf{Level 2 (10)} & \underline{0.2014} & 0.1592 & 0.1187 & 0.2003 & 0.1826 & 0.1282 & \underline{0.2048} & 0.0779 & 0.0896 & \textbf{0.2090}  \\
\multirow{-2}{*}{\textbf{DD-Sym}} & \textbf{Level 1 (5)} & 0.0597 & 0.2724 & 0.2941 & 0.0461 & \textbf{0.2301} & 0.1811 & \textbf{0.3907} & 0.0608 & 0.0027 & 0.1023  \\
& \textbf{Level 2 (10)} & \textbf{0.3187} & \textbf{0.2820} & \textbf{0.4102} & \underline{0.2320} & 0.1026 & 0.0637 & 0.1783 & 0.1409 & 0.1112 & 0.1924  \\
\multirow{-2}{*}{\textbf{DI-SIM}} & \textbf{Level 1 (5)} & 0.1017 & 0.1436 & \underline{0.3105} & 0.1153 & 0.1512 & \textbf{0.3072} & 0.0758 & 0.1284 & 0.0679 & 0.1857  \\
& \textbf{Level 2 (10)} & 0.1124 & 0.0833 & 0.0460 & 0.0482 & 0.0451 & 0.0636 & 0.1483 & 0.1128 & 0.1350 & 0.1931  \\
\multirow{-2}{*}{\textbf{Herm}} & \textbf{Level 1 (5)} & 0.1979 & 0.1442 & 0.1651 & 0.1059 & 0.0086 & \underline{0.1933} & 0.0780 & \textbf{0.1683} & 0.0395 & 0.1182  \\
& \textbf{Level 2 (10)} & 0.1124 & 0.0833 & 0.0460 & 0.0482 & 0.0451 & 0.0636 & 0.1483 & 0.1128 & \underline{0.1350} & 0.1931  \\
\multirow{-2}{*}{\textbf{Skew}} & \textbf{Level 1 (5)} & 0.1970 & 0.1866 & 0.0921 & 0.0614 & 0.0289 & 0.0944 & 0.0259 & 0.0815 & 0.0395 & 0.0356  \\
\midrule
& \textbf{Level 2 (10)} & 0.1139 & 0.0881 & 0.0573 & 0.0426 & 0.0613 & 0.0276 & 0.1369 & 0.0755 & 0.1305 & \underline{0.1951}  \\
\multirow{-2}{*}{\textbf{\myproj}} & \textbf{Level 1 (5)} & 0.0761 & \textbf{0.3215} & 0.0410 & \textbf{0.2642} & \underline{0.1939} & 0.1399 & 0.0603 & 0.0910 & \textbf{0.1759} & 0.1424  \\
\toprule
$\mathcal{F}$& \textbf{Trains (\%)}  & \textbf{0 (USL)}  & \textbf{10}  & \textbf{20}  & \textbf{30}  & \textbf{40}  & \textbf{50}  & \textbf{60}  & \textbf{70}  & \textbf{80}  & \textbf{90}  \\ \midrule
& \textbf{Level 2 (10)} & 0.0869 & 0.1268 & 0.1562 & 0.1810 & 0.4310 & 0.3964 & 0.5648 & 0.6135 & 0.7751 & 0.9032  \\
\multirow{-2}{*}{\textbf{Bi-Sym}} & \textbf{Level 1 (5)} & 0.1578 & \textbf{0.3965} & \textbf{0.6000} & 0.3334 & 0.4645 & 0.4982 & 0.5733 & \underline{0.7362} & 0.7797 & 0.8927  \\
& \textbf{Level 2 (10)} & 0.0952 & 0.1845 & \underline{0.4563} & 0.3146 & \textbf{0.5943} & 0.4519 & \underline{0.6455} & 0.7349 & 0.7531 & 0.8873  \\
\multirow{-2}{*}{\textbf{DD-Sym}} & \textbf{Level 1 (5)} & 0.2174 & 0.1822 & 0.3072 & 0.3154 & \underline{0.5200} & \textbf{0.6030} & \textbf{0.6756} & 0.6984 & 0.7929 & 0.8905  \\
& \textbf{Level 2 (10)} & 0.2255 & 0.2772 & 0.1020 & 0.3791 & 0.2689 & 0.4604 & 0.6230 & 0.6041 & 0.7315 & 0.8923  \\
\multirow{-2}{*}{\textbf{DI-SIM}} & \textbf{Level 1 (5)} & 0.2260 & 0.2911 & 0.1386 & 0.3768 & 0.3886 & \underline{0.5757} & 0.5816 & 0.7129 & 0.7619 & 0.8897  \\
& \textbf{Level 2 (10)} & 0.0970 & 0.2428 & 0.1040 & 0.3376 & 0.4176 & 0.3815 & 0.5650 & 0.6114 & 0.7730 & 0.8927  \\
\multirow{-2}{*}{\textbf{Herm}} & \textbf{Level 1 (5)} & 0.0737 & 0.2045 & 0.3102 & 0.2425 & 0.4879 & 0.4629 & 0.6180 & 0.7022 & 0.7936 & \textbf{0.9101}  \\
& \textbf{Level 2 (10)} & 0.0970 & 0.2428 & 0.1040 & 0.3376 & 0.4176 & 0.3815 & 0.5650 & 0.6114 & 0.7730 & 0.8927  \\
\multirow{-2}{*}{\textbf{Skew}} & \textbf{Level 1 (5)} & 0.0813 & \underline{0.3095} & 0.2753 & 0.2615 & 0.5009 & 0.4697 & 0.6349 & 0.6677 & \underline{0.8450} & 0.8991  \\
\midrule
& \textbf{Level 2 (10)} & \textbf{0.2440} & 0.1386 & 0.3448 & \textbf{0.4261} & 0.3839 & 0.4357 & 0.5635 & 0.5962 & 0.7933 & 0.8926  \\
\multirow{-2}{*}{\textbf{\myproj}} & \textbf{Level 1 (5)} & \underline{0.2293} & 0.2079 & 0.4524 & \underline{0.4130} & 0.4384 & 0.5050 & 0.6014 & \textbf{0.7638} & \textbf{0.8556} & \underline{0.9039}  \\
\bottomrule
\end{tabular}}
\end{table*}

\begin{table*}[htpb!]
\caption{Clustering algorithms on the Texas dataset: a comparative analysis of modularity ($\mathcal{M}$), and F-measure ($\mathcal{F}$). The best-performing model is highlighted in \textbf{bold}, and the second-best is marked with an \underline{underline}.}\label{Table:Texas}
\vspace*{1mm}
\centering
\scalebox{0.85}[0.85]{\begin{tabular}{cccccccccccc}
\toprule
$\mathcal{M}$& \textbf{Trains (\%)}  & \textbf{0 (USL)}  & \textbf{10}  & \textbf{20}  & \textbf{30}  & \textbf{40}  & \textbf{50}  & \textbf{60}  & \textbf{70}  & \textbf{80}  & \textbf{90}  \\ \midrule
& \textbf{Level 2 (10)} & 0.2244 & 0.1420 & 0.1308 & 0.1396 & 0.0791 & 0.1368 & 0.1125 & 0.0252 & 0.0166 & 0.0059  \\
\multirow{-2}{*}{\textbf{Bi-Sym}} & \textbf{Level 1 (5)} & 0.1530 & 0.1084 & 0.1091 & 0.1560 & 0.1151 & 0.0090 & 0.1340 & 0.0426 & 0.0828 & 0.1817  \\
& \textbf{Level 2 (10)} & \underline{0.3540} & 0.2166 & 0.1529 & 0.2291 & 0.0915 & \underline{0.1632} & 0.1231 & 0.0240 & 0.0086 & 0.0051  \\
\multirow{-2}{*}{\textbf{DD-Sym}} & \textbf{Level 1 (5)} & 0.2601 & \underline{0.2492} & \underline{0.2350} & \textbf{0.2553} & 0.0586 & \textbf{0.2800} & \textbf{0.3242} & \underline{0.1634} & 0.1351 & \underline{0.2755}  \\
& \textbf{Level 2 (10)} & \textbf{0.3876} & 0.2241 & \textbf{0.2727} & \underline{0.2515} & \textbf{0.2709} & 0.1580 & 0.1096 & 0.0043 & 0.0058 & 0.0017  \\
\multirow{-2}{*}{\textbf{DI-SIM}} & \textbf{Level 1 (5)} & 0.3251 & 0.2191 & 0.1324 & 0.0864 & \underline{0.1995} & 0.0180 & 0.1261 & 0.1324 & \underline{0.2356} & \textbf{0.3472}  \\
& \textbf{Level 2 (10)} & 0.1660 & 0.1336 & 0.0714 & 0.1655 & 0.0014 & 0.1378 & 0.1233 & 0.0199 & 0.0020 & 0.0071  \\
\multirow{-2}{*}{\textbf{Herm}} & \textbf{Level 1 (5)} & 0.1544 & 0.0852 & 0.2178 & 0.1955 & 0.1876 & 0.1024 & 0.0889 & 0.1577 & 0.1123 & 0.1791  \\
& \textbf{Level 2 (10)} & 0.1660 & 0.1336 & 0.0714 & 0.1655 & 0.0014 & 0.1378 & 0.1233 & 0.0199 & 0.0020 & 0.0071  \\
\multirow{-2}{*}{\textbf{Skew}} & \textbf{Level 1 (5)} & 0.1395 & 0.1052 & 0.1806 & 0.1129 & 0.1807 & 0.0175 & 0.0556 & \textbf{0.1763} & 0.0988 & 0.0089  \\
\midrule
& \textbf{Level 2 (10)} & 0.2578 & \textbf{0.2551} & 0.2076 & 0.1802 & 0.0284 & 0.1208 & 0.1055 & 0.0190 & 0.0155 & 0.0025  \\
\multirow{-2}{*}{\textbf{\myproj}} & \textbf{Level 1 (5)} & 0.1915 & 0.2154 & 0.0454 & 0.0556 & 0.0885 & 0.0719 & \underline{0.1646} & 0.1133 & \textbf{0.3336} & 0.1114  \\
\toprule
$\mathcal{F}$& \textbf{Trains (\%)}  & \textbf{0 (USL)}  & \textbf{10}  & \textbf{20}  & \textbf{30}  & \textbf{40}  & \textbf{50}  & \textbf{60}  & \textbf{70}  & \textbf{80}  & \textbf{90}  \\ \midrule
& \textbf{Level 2 (10)} & 0.1014 & 0.0298 & 0.1724 & 0.1617 & \underline{0.5147} & 0.4158 & 0.6130 & 0.6490 & 0.7542 & 0.8647  \\
\multirow{-2}{*}{\textbf{Bi-Sym}} & \textbf{Level 1 (5)} & \underline{0.2125} & 0.1344 & 0.3098 & 0.3900 & 0.3970 & 0.4951 & \underline{0.6795} & 0.6754 & 0.7830 & 0.8803  \\
& \textbf{Level 2 (10)} & 0.2015 & 0.3793 & 0.1310 & 0.1940 & 0.4068 & 0.4572 & 0.5531 & 0.6294 & 0.7733 & 0.8738  \\
\multirow{-2}{*}{\textbf{DD-Sym}} & \textbf{Level 1 (5)} & 0.1256 & 0.3345 & 0.1893 & 0.3889 & 0.4156 & \underline{0.5595} & 0.6625 & 0.7194 & 0.8194 & \underline{0.8991}  \\
& \textbf{Level 2 (10)} & \textbf{0.2398} & 0.0466 & \textbf{0.3560} & \textbf{0.5918} & 0.4634 & 0.5470 & 0.6277 & 0.6089 & 0.7315 & 0.8779  \\
\multirow{-2}{*}{\textbf{DI-SIM}} & \textbf{Level 1 (5)} & 0.1463 & 0.1822 & 0.2526 & 0.4041 & 0.3953 & 0.4580 & 0.5868 & 0.6320 & 0.7609 & 0.8897  \\
& \textbf{Level 2 (10)} & 0.0161 & 0.1508 & 0.1384 & 0.3140 & 0.4368 & 0.4704 & 0.5100 & 0.6063 & 0.7312 & 0.8673  \\
\multirow{-2}{*}{\textbf{Herm}} & \textbf{Level 1 (5)} & 0.1810 & \underline{0.4223} & 0.1118 & 0.4043 & 0.3571 & 0.4878 & 0.5548 & 0.6453 & \underline{0.8339} & 0.8905  \\
& \textbf{Level 2 (10)} & 0.0161 & 0.1508 & 0.1384 & 0.3140 & 0.4368 & 0.4704 & 0.5100 & 0.6063 & 0.7312 & 0.8673  \\
\multirow{-2}{*}{\textbf{Skew}} & \textbf{Level 1 (5)} & 0.1630 & 0.2985 & 0.1008 & \underline{0.4210} & 0.3554 & 0.4827 & 0.5575 & 0.6231 & \textbf{0.8499} & 0.8923  \\
\midrule
& \textbf{Level 2 (10)} & 0.0290 & 0.0308 & \underline{0.3290} & 0.1755 & 0.4303 & 0.4572 & 0.6234 & \textbf{0.7722} & 0.7500 & 0.8621  \\
\multirow{-2}{*}{\textbf{\myproj}} & \textbf{Level 1 (5)} & 0.0856 & \textbf{0.6265} & 0.1936 & 0.3348 & \textbf{0.7460} & \textbf{0.5920} & \textbf{0.6868} & \underline{0.7416} & 0.8055 & \textbf{0.9179}  \\
\bottomrule
\end{tabular}}
\end{table*}

\end{document}